\documentclass[10pt,a4paper]{article}

\usepackage[utf8]{inputenc}
\usepackage[T1]{fontenc}
\usepackage{lmodern}
\usepackage[english]{babel}

\usepackage[a4paper,margin=1in]{geometry}
\usepackage{setspace}
\usepackage{microtype}
\usepackage{csquotes}
\usepackage{titlesec}
\usepackage{amsmath,amssymb,amsthm,mathtools,nicefrac}
\usepackage{aliascnt}
\numberwithin{equation}{section}

\usepackage{graphicx}
\usepackage{subcaption}
\usepackage{booktabs}
\usepackage{multirow}
\usepackage{caption}
\usepackage{array}
\newcolumntype{R}[1]{>{\raggedleft\arraybackslash}p{#1}}

\usepackage{enumitem}

\usepackage{algorithm}
\usepackage{algpseudocode}

\usepackage[table]{xcolor}
\usepackage{hyperref}
\usepackage[capitalise,nameinlink,noabbrev]{cleveref}
\usepackage[numbers]{natbib}

\hypersetup{
	colorlinks=true,
	linkcolor=blue!60!black,
	citecolor=blue!60!black,
	urlcolor=blue!60!black
}
\newcommand{\gc}{\cellcolor{gray!20}}
\newtheorem{proposition}{Proposition}[section]

\newaliascnt{theorem}{proposition}
\newtheorem{theorem}[theorem]{Theorem}
\aliascntresetthe{theorem}

\newaliascnt{lemma}{proposition}

\aliascntresetthe{lemma}

\newaliascnt{corollary}{proposition}

\aliascntresetthe{corollary}

\theoremstyle{remark}

\newaliascnt{remark}{proposition}
\newtheorem{remark}[remark]{Remark}
\aliascntresetthe{remark}

\newaliascnt{assumption}{proposition}

\aliascntresetthe{assumption}

\crefname{proposition}{proposition}{propositions}
\Crefname{proposition}{Proposition}{Propositions}

\crefname{theorem}{theorem}{theorems}
\Crefname{theorem}{Theorem}{Theorems}

\crefname{lemma}{lemma}{lemmas}
\Crefname{lemma}{Lemma}{Lemmas}

\crefname{corollary}{corollary}{corollaries}
\Crefname{corollary}{Corollary}{Corollaries}

\crefname{remark}{remark}{remarks}
\Crefname{remark}{Remark}{Remarks}

\crefname{assumption}{assumption}{assumptions}
\Crefname{assumption}{Assumption}{Assumptions}

\newcommand{\R}{\mathbb{R}}

\newcommand{\locvol}{\sigma_{\mathrm{loc}}}
\newcommand{\treevol}{\sigma^{\mathrm{tree}}}
\newcommand{\gprvol}{\widehat{\sigma}^{\mathrm{GPR}}}

\newcommand{\truelv}{\sigma_{\mathrm{LV}}}
\newcommand{\Tree}{\mathbf{S}}

\title{\textbf{
Neural Calibration of a Complete Market Model
}}
\author{
	Andrea Molent\thanks{Dipartimento di Scienze Economiche e Statistiche, Università degli Studi di Udine, Udine, Italy, andrea.molent@uniud.it.}
	\and
	Michel Vellekoop\thanks{Faculty of Economics and Business, University of Amsterdam, Amsterdam, the Netherlands, m.h.vellekoop@uva.nl.}
}

\date{}

\usepackage{placeins}
\usepackage{xfp}
\usepackage{siunitx}

\newcommand{\nz}[1]{%
	\begingroup
	\sisetup{round-mode=places, round-precision=0}%
	$\num{#1}$%
	\endgroup
}

\newcommand{\no}[1]{%
	\begingroup
	\sisetup{round-mode=places, round-precision=1}%
	$\num{#1}$%
	\endgroup
}

\newcommand{\nm}[1]{%
	\begingroup
	\if\relax\detokenize{#1}\relax
	\else
	\sisetup{round-mode=places, round-precision=0}%
	\num{\fpeval{#1}}%
	\fi
	\endgroup
}

\newcommand{\nd}[1]{%
	\begingroup
	\sisetup{round-mode=places, round-precision=2}%
	\ensuremath{\num{#1}}%
	\endgroup
}

\newcommand{\nl}[1]{%
	\begingroup
	\if\relax\detokenize{#1}\relax
	\else
	\sisetup{round-mode=places, round-precision=1}%
	\num{\fpeval{#1*0.001}}%
	\fi
	\endgroup
}

\begin{document}
	
	\maketitle
	
	\begin{abstract}
		We propose a neural calibration method to construct a recombining binomial tree directly from a set of given option prices. Rather than estimating a continuous option pricing function or a local volatility surface as an intermediate object, a neural network is used to deform a benchmark lattice. This leads to a discrete pricing model which is guaranteed to be arbitrage-free, complete, easy to interpret, and can be used directly for pricing and to find replicating trading strategies.
        Calibration is formulated as a penalized optimization problem that combines a repricing error with an admissibility penalty, and an optional  spatial regularization term based on implied local volatilities.
        
            Numerical experiments on synthetic and SPX market data show that the proposed approach
		yields accurate repricing and is very competitive when compared to recently proposed other neural calibration methods. It preserves
		the computational advantages of lattice-based valuation and hedging. In particular, the calibrated tree can be reused
		to price contracts that allow early exercise, and could even be calibrated directly with American option prices.
	\end{abstract}
	

    
	\section{Introduction}
	
	Calibrating an option pricing model to market data remains a central problem in computational finance. A large part of the existing literature focuses on constructing a continuous relationship, such as an implied-volatility surface, a local-volatility surface, or an option-price surface, for different strikes and maturities. Once such a relationship has been estimated, additional numerical procedures are typically required in order to price new contracts. Machine learning methods can help to generate better fits to the data, but generated pricing networks may be difficult to interpret, and replicating strategies or sensitivities such as the "Greeks" may not be directly available, especially when options with the possibility of early exercise need to be priced using models that have been calibrated on options without that possibility. 
	
	We therefore choose a different approach. Instead of reconstructing a continuous implied relationship and subsequently using it for valuation, we calibrate a recombining binomial tree directly to observed option prices. The output of the procedure is therefore not merely an interpolated surface, but a discrete market model that is immediately available for pricing and hedging of European and American options. 
	
	The construction starts from a base tree; this may be, for example, a Cox--Ross--Rubinstein (CRR) tree \citep{cox1979} or a previously calibrated tree. A neural network then deforms the node values of this benchmark tree in log-space. The transformed tree induces state-dependent risk-neutral transition probabilities, from which option prices are computed and matched to market observations through a differentiable loss function. In addition to the pricing error, this objective function includes a penalty to prevent inadmissible transition probabilities and an optional spatial regularizer that smooths the tree-implied local volatility.
	%When the spatial regularizer is not used, the  calibration choices are limited to the number of time steps \(N_T\) and to the maximum number of epochs  \(N_E\) that are used to train the neural network.
	
	Existing binomial-tree and implied-tree approaches typically construct a lattice through forward methods or through ad hoc calibration rules.  \citet{dermankani1994} first introduced a seminal forward-induction construction that extracts a smile-consistent binomial tree from European option prices. Closely related in spirit, \citet{rubinstein1994} derived implied recombining binomial trees from observed option prices and shows how such trees can be used to recover a fully specified discrete pricing structure. The  no-arbitrage condition  is a central issue in such constructions: in particular, \citet{moriggia2009} revisit the Derman--Kani framework and point out that negative probabilities may be generated in option-implied trees.
	The constant-probability recombining tree proposed by \citet{li2000new} offers an alternative to earlier implied-tree algorithms, but the theoretical support for the method is based on local consistency and weak-convergence arguments, rather than on a general pricing-convergence result for implied trees. Our experiments suggest that, while useful as an approximation, it may encounter difficulties in producing prices that converge to the continuous time limit of the local volatility models in discrete time.
	
	A second strand of the literature concerns local-volatility calibration. Based on theoretical results in \citet{dupire1994}, these approaches aim to recover a continuous local-volatility surface from option prices, typically through regularization and smoothing \cite{crepey2003}. The methods   often require numerical solutions for PDE's to price new payoffs, and calibration may be difficult and time-consuming. More recently, neural-network methods have been proposed for option-price interpolation and local-volatility calibration under no-arbitrage constraints. \citet{chataigner2020} develop a deep-learning framework for interpolating European vanilla option prices while recovering a local-volatility surface. 
    In a related direction, \citet{xian2026risk} propose  generative networks that learn maturity-dependent risk-neutral densities directly from option prices, without specifying a complete dynamic model for the underlying asset. In 
    \citet{wang2025} other coordinates are used: not the densities but a local volatility surface is calibrated, in a procedure in which option prices and local volatilities are learned jointly.   	

     Other related contributions use machine learning or regularization for option calibration and pricing in different ways. \citet{fan2026} parameterize the drift and volatility functions of a stochastic differential equation for stock prices using neural networks and calibrate them to market option prices, using simulation-based methods for European options and PDE-based methods for American options. \citet{jang2019} propose a generative Bayesian neural-network model for American index options, while \citet{zhang2023} construct implied-volatility surfaces that change over time, using deep neural networks subject to static no-arbitrage constraints.  \citet{geng2014} formulate local-volatility calibration as a non-parametric inverse problem using second-order Tikhonov regularization.
     In most approaches, the dynamics of stock prices are assumed to be Markovian; for machine-learning methods without this assumption, see for example \citet{goudenege2020}.

   Rather than estimating densities, or continuous surfaces for prices or local volatilities, we calibrate a recombining binomial market model directly from option prices. The paper is therefore closer in spirit to the implied-tree tradition initiated by Derman and Kani, but we replace classical forward-induction procedures by a neural deformation method and an optimization-based regularization strategy. Accordingly, the calibration problem aims to pick the best choice in a  set of admissible recombining binomial trees, and its solution yields a discrete, market-consistent, arbitrage-free, and complete pricing model. Since our method to create more regular volatility surfaces does not use the Dupire equation for European options in local volatility models, we can also handle American options which do not satisfy that equation.

   We compare the proposed approach to recent neural local-volatility methods and show that it is effective in reproducing market prices while maintaining a plausible local-volatility profile.
	In particular, we show this for a test of our method on a real dataset of SPX options.
	
	The remainder of the paper is organized as follows. \Cref{sec:model} introduces the discrete time model for the arbitrage-free and complete binomial market. \Cref{sec:method} presents the neural parametrization, pricing procedure, and optimization criteria. \Cref{sec:theory} establishes mathematical properties of the class of admissible trees and the penalized calibration problem. \Cref{sec:experiments} reports numerical results for several case studies. Finally, \Cref{sec:conclusion} summarizes the main findings.
	
	\section{An Arbitrage-Free and Complete Market Model}
	\label{sec:model}
	
	Consider a uniform time grid
	\[
	0=t_0<t_1<\cdots<t_{N_T}=T_{\max},
	\qquad \Delta t=t_{n+1}-t_n,
	\]
    with
	\(
	T_{\max}=\max_{1\le i\le M} T_i,
	\)
    the maximum over maturity dates $T_i$ in the set of options we want to fit during the calibration.
		Let $S_{n,j}$ denote the asset price at time $t_n$ under scenario $j$ at that time, with $j=0,\dots,n$ and $n=0,\dots,N_T$. We assume a recombining binomial structure, so that from node $(n,j)$ the process can move to either $(n+1,j)$ or $(n+1,j+1)$.
	   Trees are indexed according to the standard
	convention,
	\[
	S_{n,0}<S_{n,1}<\cdots<S_{n,n}, \qquad n=1,\ldots,N_T,
	\]
	so if the stock price equals $S_{n,j}$ then \(S_{n+1,j}\) is the lower
	successor and \(S_{n+1,j+1}\)  the upper successor.
	Under deterministic but possibly time-varying interest rates $r_n$ and dividend yields $q_n$ over the time interval $[t_n,t_{n+1}]$, the one-step risk-neutral pricing relation for a contingent claim value $V_{n,j}$ is
	\begin{equation}
		V_{n,j}
		=
		e^{-r_n\Delta t}
		\Bigl[
		(1-p_{n,j})V_{n+1,j}
		+
		p_{n,j}V_{n+1,j+1}
		\Bigr],
	\end{equation}
	where $p_{n,j}$ denotes the  risk-neutral  probability that the stock price moves upwards  in node $(n,j)$.  
	Imposing the martingale condition on the discounted asset price yields
	\begin{equation}
		\label{eq:rnp}
		p_{n,j}
		=
		\frac{e^{(r_n-q_n)\Delta t}S_{n,j}-S_{n+1,j}}
		{S_{n+1,j+1}-S_{n+1,j}}.
	\end{equation}
	
	The following proposition summarizes the
	key admissibility condition; the proof is given in Appendix~\ref{app:proof-arbfree}.
	
	\begin{proposition}
		\label{prop:arbfree}
		Assume that for every node $(n,j)$,
		\begin{equation}
		    S_{n+1,j}
		<
		e^{(r_n-q_n)\Delta t}S_{n,j}
		<
		S_{n+1,j+1}.\label{eq:pop21}
		\end{equation}
		Then the  corresponding recombining binomial model is arbitrage-free and the market is dynamically complete.
	\end{proposition}

 \medskip

	Our neural calibration does not impose the inequalities in
	\Cref{prop:arbfree} as hard constraints at each optimization step. Rather,
	the calibration is initialized from an admissible benchmark tree, and the objective contains a term which ensures that violations of the inequalities in \Cref{prop:arbfree}
     are penalized. The calibrated trees generated in our numerical
	experiments are checked ex post to satisfy \eqref{eq:pop21}. Hence, it applies for all nodes in all reported
	calibrated trees.

	\section{Neural Tree Calibration}
	\label{sec:method}

		\subsection{Benchmark tree and neural deformation}
	Let $S^{0}_{n,j}$	denote a benchmark recombining tree such as a standard Cox-Ross-Rubinstein lattice \cite{cox1979} or a previously calibrated tree.
	Rather than learning node values from scratch, we let a neural network produce a deformation of this benchmark tree in log-space. More precisely, for each node $(n,j)$ we define
	\begin{equation}
		\delta_{n,j}
		=
		f_{\theta}\!\left(
		\frac{t_n}{T_{\max}},
		\log\frac{S^{0}_{n,j}}{S_0}
		\right),\label{eq:31}
	\end{equation}
	where $f_{\theta}$ is a feedforward neural network with parameters $\theta$ and $S_0$ is the initial asset price.
	The calibrated tree is then given by
	\begin{equation}
		\label{eq:tree-transform}
		S^{(\theta)}_{n,j}
		=
		S^{0}_{n,j}\exp(\delta_{n,j}),
	\end{equation}
	so positivity of node values is guaranteed.

	The feedforward neural network $f_{\theta}$ is a parametric map obtained by composing affine transformations and nonlinear activation functions; see, for example, \citet[Chapter~6]{goodfellow2016deep}. For an input \(x\in\R^{d_0}\), a network with \(L-1\) hidden layers is defined recursively by
	\[
	h^{(0)}=x,\qquad
	h^{(\ell)}=\phi\!\left(A^{(\ell)}h^{(\ell-1)}+b^{(\ell)}\right),
	\quad \ell=1,\dots,L-1,\qquad
	f_{\theta}(x)=A^{(L)}h^{(L-1)}+b^{(L)},
	\]
	where
	\(
	\theta=\{A^{(\ell)},b^{(\ell)}\}_{\ell=1}^L
	\)
	collects all weight matrices and biases, and \(\phi\) is an activation function, for example the rectified linear unit (ReLU) function $\phi(x)=\max\{ x,0\}=:(x)^+$ on $\mathbb R$.
	
	In the present paper, the input dimension is \(d_0=2\)
	and the scalar output defines the node-wise log-deformation
	\(	\delta_{n,j}=f_{\theta}(x_{n,j})\) in \eqref{eq:31}.
	Thus, rather than optimizing all node values independently, we optimize a shared parametric map $f_\theta:\R^2\to\R$ over the whole lattice.
	This map provides a flexible non-linear function class which allows us to avoid a high-dimensional unconstrained search over node values. Multilayer feed-forward networks are known to be rich approximation classes on compact sets; see \citet{hornik1989multilayer}. This makes them well suited to capture non-trivial deformations of the benchmark lattice induced by the option data.

\subsection{Pricing on the calibrated tree}
\label{sec:pricing-tree}

Given the node values \(S_{n,j}^{(\theta)}\), the risk-neutral probabilities
\(p_{n,j}^{(\theta)}\) are computed through \Cref{eq:rnp}. European
option prices are obtained by propagating Arrow--Debreu state prices across the
tree and aggregating the corresponding payoffs at the relevant maturity layer;
see Arrow~\citep{arrow1964role}, Debreu~\citep{debreu1959}. American-style claims  
are priced by backward induction with the early-exercise constraint.

Let \(\lambda_{n,j}^{(\theta)}\) denote the Arrow-Debreu state price at node \((n,j)\). These quantities are   propagated recursively according to, for \(j=1,\dots,n\),
\[
\lambda_{n+1,j}^{(\theta)}
=
e^{-r_n\Delta t}(1-p_{n,j}^{(\theta)})\,\lambda_{n,j}^{(\theta)}
+
e^{-r_n\Delta t}p_{n,j-1}^{(\theta)}\,\lambda_{n,j-1}^{(\theta)},\qquad \lambda_{0,0}^{(\theta)}=1,
\]
while at the boundaries ($j=0$ or $j=n+1$), the recursion reduces to
\[
\lambda_{n+1,0}^{(\theta)}
=
e^{-r_n\Delta t}(1-p_{n,0}^{(\theta)})\,\lambda_{n,0}^{(\theta)},
\qquad
\lambda_{n+1,n+1}^{(\theta)}
=
e^{-r_n\Delta t}p_{n,n}^{(\theta)}\,\lambda_{n,n}^{(\theta)}.
\]
For a European call or put with strike \(K\) and maturity \(t_n\) we thus find the price
\[
\Pi^{(\theta)}(K,t_n)
=
\sum_{j=0}^n \lambda_{n,j}^{(\theta)}
g(S_{n,j}^{(\theta)}),
\]
with  $
g(S)=(S-K)^+
$ or $
g(S)=(K-S)^+$,
and if a maturity \(T\in(t_n,t_{n+1})\) does not coincide with a grid point, the price can be defined by linear interpolation:
\[
\Pi^{(\theta)}(K,T)
=
(1-w)\Pi^{(\theta)}(K,t_n)+w\Pi^{(\theta)}(K,t_{n+1}),
\qquad
w=\frac{T-t_n}{t_{n+1}-t_n}.
\]
This interpolation is used for matching off-grid quoted maturities; the arbitrage-free and completeness statements refer to the discrete-time market defined on the tree dates\footnote{The method presented here can easily be extended to grids that are not uniform in time, but to avoid the more cumbersome notation we have not implemented that here.}.

One can price American-style contracts by
standard backward induction. If the maturity is
\(t_m\) and \(g\) denotes the payoff function, then
$
V^{(\theta)}_{m,j}=
g\!\left(S^{(\theta)}_{m,j}\right)
$
for 
$j=0,\ldots,m,
$
and, for \(n=m-1,\ldots,0\),
\[
V^{(\theta)}_{n,j}
=
\max\left\{
g\!\left(S^{(\theta)}_{n,j}\right),
e^{-r_n\Delta t}
\left[
\left(1-p^{(\theta)}_{n,j}\right)
V^{(\theta)}_{n+1,j}
+
p^{(\theta)}_{n,j}
V^{(\theta)}_{n+1,j+1}
\right]
\right\},
\qquad j=0,\ldots,n,
\]
The American option price generated by the model is then  \(\Pi^{(\theta)}(K,t_m)= V^{(\theta)}_{0,0} \). For maturities
falling between two tree dates, the same linear interpolation convention as
above is applied to the corresponding tree prices.

Let
\[
\{(K_i,T_i,\Pi_i^{\mathrm{mkt}},g_i,\mathrm{style}_i)\}_{i=1}^{N_O}
\]
denote the set of $N_O$ observed option quotes, where \(K_i\) is the strike, \(T_i\) is the
maturity, \(\Pi_i^{\mathrm{mkt}}\) is the observed market price, \(\mathrm{style}_i\) specifies the exercise style (European or American), and  \(g_i\) is the payoff
function, so for standard calls and puts 
$
g_i(S)=(S-K_i)^+
$ or $
g_i(S)=(K_i-S)^+$, respectively.
For each quote, \(\Pi_i^{(\theta)}\) denotes the corresponding model price
computed on the calibrated tree by the pricing procedure described above.

When calibrating, we minimize the mean squared pricing
error
\begin{equation}
	\label{eq:mse}
	\mathrm{MSE}(\theta)
=
\frac1{N_O}
\sum_{i=1}^{N_O}
\left(
\Pi_i^{(\theta)}
-
\Pi_i^{\mathrm{mkt}}
\right)^2,
\end{equation}
	under the constraint that all risk-neutral probabilities remain admissible. During the optimization, we therefore  discourage violations of admissibility  by  a penalty term:
	\begin{equation}\label{eq:Pprob}
		\mathcal{P}_{\mathrm{prob}}(\theta)
		=
		\sum_{n,j}
		\left[
		\phi\!\bigl(-p_{n,j}^{(\theta)}\bigr)^2
		+
		\phi\!\bigl(p_{n,j}^{(\theta)}-1\bigr)^2
		\right],
	\end{equation}
    with $\phi(x)=\max\{ x,0\}=:(x)^+$ the ReLU function we defined earlier.
	In the implementation, $\mathcal{P}_{\mathrm{prob}}(\theta)$ is multiplied by a large constant $\lambda_{\mathrm{prob}}$, so that
	inadmissible transition probabilities become prohibitively expensive during
	optimization.	
    
	In the numerical experiments reported in this paper, this penalization
	mechanism proved to be very effective. In all calibrated trees, we verified ex
	post both probability admissibility
	and the one-step node-ordering inequalities
	at every node. These conditions were always satisfied, so by
	\Cref{prop:arbfree}, all generated binomial markets were arbitrage-free and complete.

\subsection{Tree-implied local volatility and spatial regularization}
\label{sec:tree-lv-sr}

Although our method does not calibrate a continuous local-volatility surface
directly, each one-step transition of the calibrated tree induces a local
conditional variance for the log-return. This 
quantity is the basis for our spatial regularization term.
Let
\[
x^{(\theta)}_{n,j}
:=
\log\left(\frac{S^{(\theta)}_{n,j}}{S_0}\right)
\]
denote the log-price coordinate of node \((n,j)\), and let
\(\hat p^{(\theta)}_{n,j}\) be the clipped version of the corresponding local risk-neutral up probability:
\[
{\hat p}^{(\theta)}_{n,j}
=
\min\left\{1,\max\left\{0,p^{(\theta)}_{n,j}\right\}\right\}.
\]

We define the tree-implied local volatility at node \((n,j)\), denoted by
\({\treevol_{n,j}}^{(\theta)}\), through the one-step log-dispersion of the
calibrated tree and use \(v^{(\theta)}_{n,j}\) for the corresponding squared local volatility:
\begin{equation}
	\label{eq:tree-lv}
	v^{(\theta)}_{n,j}
	:=
	\left({\treevol_{n,j}}^{(\theta)}\right)^2
	=
	\frac{
		\hat p^{(\theta)}_{n,j}
		\left(1-\hat p^{(\theta)}_{n,j}\right)
		\left(
		x^{(\theta)}_{n+1,j+1}
		-
		x^{(\theta)}_{n+1,j}
		\right)^2
	}{
		\Delta t
	}.
\end{equation}

To discourage  spatial oscillations of the tree-implied
local-volatility surface, we add a regularization term along the log-price
direction.
The penalty is applied to \(v_{n,j}^{(\theta)}\), that is, to the squared
tree-implied local volatility, rather than directly to
\({\treevol_{n,j}}^{(\theta)}\). This choice avoids the
repeated evaluation of square roots and their derivatives during training,
thereby helping to reduce the computational cost of the regularization term.

For each penalized layer \(n\), define
\[
\Delta x^{(\theta)}_{n,j}
=
x^{(\theta)}_{n,j+1}
-
x^{(\theta)}_{n,j},
\qquad 
\Delta v^{(\theta)}_{n,j}
=
v^{(\theta)}_{n,j+1}
-
v^{(\theta)}_{n,j}
\qquad
j=0,\ldots,n-1.
\]
Moreover, let
\[
L_n^{(\theta)}
=
x^{(\theta)}_{n,n}
-
x^{(\theta)}_{n,0}
\]
be the log-price width of layer \(n\). The first two layers are not penalized, so that the spatial penalty is averaged
over the set $\mathcal N=\{2,\ldots,N_T-1\}$.
Given a small numerical constant \(\varepsilon_x>0\), the %
spatial regularization penalty is then defined as
\begin{equation}
	\label{eq:Pspace}
	\mathcal{P}_{\mathrm{space}}(\theta)
	=
	\frac{1}{|\mathcal N|}
	\sum_{n\in\mathcal N}
	\mathcal{P}_{\mathrm{space},n}(\theta),\qquad
     \mathcal{P}_{\mathrm{space},n}(\theta)
=
\frac{1}{L_n^{(\theta)}+\varepsilon_x}
\sum_{j=0}^{n-1}
\frac{
	\left(
	\Delta v^{(\theta)}_{n,j}
	\right)^2
}{
	\Delta x^{(\theta)}_{n,j}+\varepsilon_x}.
\end{equation}

This expression may be viewed as a normalized
discrete first-order Tikhonov penalty for the squared tree-implied local
volatility along the log-price direction. Since it represents the local
variance associated with the one-step log-return, we refer to it as a
local-volatility or spatial regularization term. 

\subsection{Calibration objective}
\label{sec:complete-loss}

Combining the pricing error, the probability-admissibility penalty, and the
local-volatility regularization term, the objective function for the calibration is
\begin{equation}
	\label{eq:complete-objective}
	\mathcal L(\theta)
	=
	\mathrm{MSE}(\theta)
	+
	\lambda_{\mathrm{prob}}\mathcal{P}_{\mathrm{prob}}(\theta)
	+
	\lambda_{\mathrm{space}}\mathcal{P}_{\mathrm{space}}(\theta),
\end{equation}
where \(\lambda_{\mathrm{prob}}\ge 0\) and \(\lambda_{\mathrm{space}}\ge 0\)
control the strength of the two penalty terms. In the numerical implementation,
\(\lambda_{\mathrm{prob}}\) is fixed at a large value in order to strongly
discourage inadmissible transition probabilities and is not treated as a tuning
parameter. The coefficient \(\lambda_{\mathrm{space}}\) is optional: the default
calibration corresponds to \(\lambda_{\mathrm{space}}=0\), while positive values
smoothen the local volatility surface
when the option data are less informative.
In applications, \(\lambda_{\mathrm{space}}\) should be chosen conservatively, in line with the standard validation-based use of regularization parameters in ill-posed inverse problems and smoothing methods \citep{crepey2003,tikhonov1977}. 

The calibrated parameters are therefore obtained by solving
\[
\theta^\star
\in
\operatorname*{arg\,min}_{\theta}
\mathcal L(\theta),
\]
up to the accuracy reached by the numerical optimizer. The calibrated tree is
then defined as \(S^{(\theta^\star)}\) using \eqref{eq:31}-\eqref{eq:tree-transform}, which generates the corresponding risk-neutral
probabilities \(p^{(\theta^\star)}\) in \eqref{eq:rnp}.

\label{sec:gpr-local-vol}

Although the proposed method calibrates a discrete recombining tree rather than
a continuous local-volatility model, reconstructing a continuous
local-volatility surface from the calibrated tree remains useful.
The discrete dataset used in the reconstruction is
\[
\mathcal{D}^{(\theta)}
=
\left\{
\bigl(t_n,S_{n,j}^{(\theta)},{\treevol_{n,j}}^{(\theta)}\bigr)
:\; 0 \le j \le n < N_T
\right\},
\]
with the \({\treevol_{n,j}}^{(\theta)}\) as defined in \eqref{eq:tree-lv}.
Since these quantities are available only at irregularly spaced lattice
locations, we reconstruct from them a continuous function in the whole \((t,S)\)-domain. To this end, we fit a Gaussian process regression (GPR)
model to the inputs
\((t_n,S_{n,j}^{(\theta)})\) and responses
\({\treevol_{n,j}}^{(\theta)}\) collected in $\mathcal{D}^{(\theta)}$, and define the reconstructed surface as the
posterior mean, which we denote by
\[
\gprvol{}^{(\theta)}(t,S).
\]
The regression uses a squared-exponential kernel with automatic relevance
determination and a constant basis function, see  \citet{rasmussen2006gpml} for details.
The degree
of smoothness can be adjusted through the noise level used in the fit: larger
values produce a smoother surface, whereas smaller values make the
reconstruction track the tree values more closely. 
This smoothing step does not affect the calibrated tree itself, but it provides an a posteriori estimate of the
 local volatility structure implied by the calibration.

\section{Well-posedness of the Calibration Problem}
\label{sec:theory}

In this section we formalize the penalized calibration as a finite-dimensional optimization problem over a class of recombining trees with a fixed time grid. A recombining binomial tree with a fixed number of \(N_T\) time steps contains \(d=(N_T+1)(N_T+2)/2\) nodes and can be identified with a vector $(S_{0,0}^{(\theta)},S_{1,0}^{(\theta)},S_{1,1}^{(\theta)},S_{2,0}^{(\theta)},\ldots,S_{N_T,N_T}^{(\theta)})$ in \(\mathbb{R}^d\).

%\subsection{Well-posedness of the finite-dimensional penalized problem }

For fixed constants \(0<m<M\) and \(\eta>0\), let \(\mathcal{A}_{\eta}(m,M)\subset \mathbb{R}^q\) denote the set of all values $\theta$ that lead to recombining trees 
\(
\Tree^{(\theta)}=
%\{S_{n,j}^{(\theta)}\}_{0\le j\le n\le N_T}
(S_{0,0}^{(\theta)},S_{1,0}^{(\theta)},S_{1,1}^{(\theta)},S_{2,0}^{(\theta)},\ldots,S_{N_T,N_T}^{(\theta)})
\)
such\footnote{We sometimes use the shorthand notation \(\{S_{n,j}^{(\theta)}\}_{0\le j\le n\le N_T}\) here to denote all tree values, but they should be thought of as the indicated vector which is ordered, first in $n$ and then in $j$.} that
\begin{align}
	m &\le S_{n,j}^{(\theta)} \le M, 
	\qquad 0\le j\le n\le N_T, 
	\label{eq:Aeta-bounds}\\
	S_{n+1,j}^{(\theta)}+\eta &\le e^{(r_n-q_n)\Delta t}S_{n,j}^{(\theta)} \le S_{n+1,j+1}^{(\theta)}-\eta,
	\qquad 0\le j\le n\le N_T-1.
	\label{eq:Aeta-margin}
\end{align}
Thus any \(\theta\in \mathcal{A}_{\eta}(m,M)\) generates a positive recombining tree satisfying the one-step no-arbitrage inequalities with a uniform margin $\eta$, and for each option quote \(i\) a corresponding price $\Pi_i^{(\theta)}$ as defined in subsection \ref{sec:pricing-tree}.

\begin{theorem}
	\label{thm:wellposedness}
	Fix \(N_T\in\mathbb{N}^*\), 
    \(\lambda_{\mathrm{prob}}\geq 0\)
		and \(\lambda_{\mathrm{space}}\geq 0\), \(0<m<M\), and \(\eta>0\). 
        Assume that the neural parametrization
    \( (\theta,x,y)\mapsto f_\theta(x,y)\)
	is continuous and that there exists a $\theta\in\mathbb{R}$ such that for any $(x,y)\in ([0,1]\times\mathbb{R})$,  $f_\theta(x,y)=0$. 
    Then 
    the penalized objective
		\(\mathcal L\) defined in \eqref{eq:complete-objective} admits at least one
		global minimizer on any non-empty compact subset of \(\mathcal{A}_{\eta}(m,M)\).
\end{theorem}

The proof is given in Appendix~\ref{app:proof-wellposedness}. It combines the compactness of the admissible class of trees, the stability of strict admissibility under sufficiently small log-deformations, and the continuity of the tree-to-price and tree-to-penalty maps.

\begin{remark}
	\label{rem:nonconvexity}
	The theorem guarantees existence of global minimizers for the finite-dimensional penalized problem, but it does not imply convergence of the neural training algorithm to such minimizers. Indeed, optimization is performed in parameter space through the  non-linear map \(\theta\mapsto\Tree^{(\theta)}\), and the resulting objective function need not be convex. 
\end{remark}

\begin{remark}
	\label{rem:diffusion-limit}
	The previous theorem is concerned with calibration for a given set  of option data. If these data are consistent with a certain local volatility model in continuous time, it is natural to ask whether a sequence of calibrated trees based on more and more option data converges, as \(N_T\to\infty\), to this continuous-time local-volatility model if tree design quantities may be allowed to depend on \(N_T\), for instance  \(m(N_T)\), \(M(N_T)\), and \(\eta(N_T)\). Such a conclusion would require additional assumptions of the kind commonly used in weak-approximation results. More precisely, if the discrete log-price increments vanish uniformly on compacts, the associated conditional first and second moments converge to limiting coefficients \(b\) and \(a\), the corresponding third-moment remainder is negligible, and the induced log-price processes are tight in \(D([0,T];\mathbb{R})\), then the piecewise-constant log-price processes converge weakly to the diffusion
	\[
	dX_t=b(t,X_t)\,dt+\sqrt{a(t,X_t)}\,dW_t,
	\qquad X_0=\log S_0,
	\]
    with
    \[
	b(t,x)=r(t)-q(t)-\frac12 \sigma_{\mathrm{loc}}^2(t,e^x),
	\qquad
	a(t,x)=\sigma_{\mathrm{loc}}^2(t,e^x),
	\]
	provided that the martingale problem for the limiting generator is well posed. Under the usual additional uniform-integrability condition, this also yields convergence of discounted expectations for continuous payoffs with polynomial growth, and thus for European options with such payoffs.

In practice, we are interested in the proposed calibration method for a finite number of option data. Our calibration always involves a well-posed optimization
problem over a set of strictly admissible recombining trees. The
diffusion-limit perspective may provide additional intuition, but it is not
needed for the validity or practical usefulness of the option pricing and hedging method.     
\end{remark}

	\section{Numerical Experiments}
	\label{sec:experiments}
	
	\subsection{Experimental design}
	
	Our numerical study includes both synthetic and market-based experiments. Market-based data are collected for the SPX index, and in the synthetic setting option prices are generated under a known local-volatility specification, so we can assess both pricing accuracy and the  ability of the method to recover a given local-volatility profile.
	
     Calibration is performed through a single continuous training run
	starting from the initial CRR tree \citep{cox1979} with the constant volatility value that minimizes the RMSE over all option data. The neural network has two hidden and fully connected layers with $N_L$ ReLU units, and is initialized so that
	the initial deformation is zero.
	Training is performed with Adam; see, for example, \citet{kingma2015adam}. At
	the beginning of the run, a learning-rate finder is used to select a suitable
	initial learning-rate scale. The learning rate is then evolved according to a
	two-regime cyclical cosine schedule. During the initial part of the run, shorter
	cycles and a larger maximum learning rate are used to encourage exploration.
	During the refinement part, longer cycles and a smaller maximum learning rate
	are used, and gradients are clipped component-wise for numerical stability.

	Violations of the one-step risk-neutral probability constraints are discouraged
	through the quadratic probability-admissibility penalty described above, with
	the fixed scaling coefficient \(10^6\). Since training starts from the ordered
	CRR benchmark with zero initial deformation, the optimization is initialized
	inside the ordered admissible tree class. Probability admissibility and node
	ordering are then verified ex post on the calibrated trees. 
	
	In the synthetic experiments, the total training horizon is \(8000\) epochs.
	During the run, the best learnable parameters encountered so far are stored. At
	the checkpoint epochs
	$\
	N_E\in\{1000,2000,4000,8000\},
	$
	the best-so-far calibrated tree is saved without interrupting training.
	In the SPX experiment, the same continuous training protocol is used, but the
	training horizon is extended to \(16000\) epochs, with checkpoints
	$	N_E\in\{2000,4000,8000,16000\}.
	$
	In practical applications, the prescribed budget can also be interpreted as an upper bound: the run may be stopped earlier if the monitored pricing errors or validation diagnostics have plateaued.

    Unless otherwise stated, option-pricing RMSEs are reported in price units and
	scaled by a factor $10^2$. Thus, an entry equal to $x$ in a row labelled
	$(\times 10^{-2})$ corresponds to an unscaled pricing RMSE equal to
	$x\times10^{-2}$. Local-volatility errors are relative RMSEs and are reported
	as a percentage, following the convention used by Wang et al.~\cite{wang2025} in their paper\footnote{
	All computations
	reported in this paper were implemented in MATLAB  on a CPU for a desktop machine equipped with
	an Intel Core i7-12700 processor (2.10\,GHz), 32\,GB of RAM, and Windows 11 Pro.
	The computational time reported for the method of \citet{wang2025} was estimated
	using the Python code made available by the authors. This code was also run on a
	CPU, due to the lack of a GPU suitable for their precise computational settings.
	The runtime comparisons should therefore be interpreted as indicative
	rather than a strictly hardware-optimized benchmark.

}.

	\subsection{Synthetic-data experiments}\label{sec:sy_da_ex}

	We first consider synthetic option markets generated from the   local-volatility function used in \citet{wang2025}, 
	\begin{equation}
		\label{eq:lv}
		\truelv(t,s)=0.3 + y e^{-y},
		\qquad
		y=(t+0.1)\sqrt{\frac{s}{S_0}+0.1},
	\end{equation}
    with initial stock price $S_0=1000$, risk-free rate $r=0.04$ and zero dividends, $q=0$.

The goal is twofold: to evaluate how accurately the calibrated tree reproduces option prices and to verify whether the induced local-volatility structure recovers \eqref{eq:lv}.
To facilitate comparison with \citet{wang2025}, we use the same strikes and maturities. The European call prices are computed by solving the option-pricing PDE under the local-volatility model in \eqref{eq:lv} using 1000 space steps and 500 time steps. The first configuration, denoted by $10\times 20$, uses a grid of European call prices with $N_{T}^{\text{opt}}=10$ maturities uniformly distributed between $0.3$ and $1.5$ and $N_{K}^{\text{opt}}=20$ strikes uniformly distributed between $500$ and $3000$. The second one, denoted by $3\times 6$, is defined analogously, with $N_{T}^{\text{opt}}=3$ maturities and $N_{K}^{\text{opt}}=6$ strikes over the same ranges. 
	For out-of-sample pricing evaluation, we use a dense $256\times256$ grid with
256 maturities uniformly distributed over $[0.3,1.5]$ and 256 strikes uniformly
distributed over $[500,3000]$. This evaluation grid is defined over the same
domain as the training grids.
	
		\subsubsection{Option Pricing}

	In our numerical experiments, neither Li's tree \cite{li2000new} nor the binomial tree of Moriggia et al. \cite{moriggia2009} displayed convergence towards the PDE benchmark. Table~\ref{tab:Li_Moriggia} shows that, although both constructions remain implementable for finer time discretizations, the resulting European call prices do not systematically approach the benchmark values; in some cases, the pricing errors may even deteriorate at finer discretizations. 
	
		This 
    shows a limitation with respect to the specific refinement criterion considered here, but should not obscure the distinct contributions of the two approaches. Li's construction is mainly motivated by weak-approximation considerations for the tree dynamics under suitable regularity assumptions, while Moriggia et al.\ focus on admissibility and empirical pricing performance.
    It shows how delicate the construction of a recombining binomial tree remains even when the local-volatility specification is known, and motivates our search for a different calibration procedure.

	We now evaluate the   pricing performance of our proposed neural calibration method.
  \Cref{tab:Price_convergence} reports the relative percentage errors with respect to the PDE benchmark for the two training configurations, $10\times 20$ and $3\times 6$, after a cumulative training budget of $N_E=8000$ epochs. In this experiment $\lambda_{\mathrm{space}}=0$, so that the reported results isolate the fitting capability of the procedure, without considering the regularity of calibrated values on the tree.
	
	\begin{table}[b!]
	\begin{centering}
		\small{ 
			\begin{tabular*}{\textwidth}{@{\extracolsep{\fill}}l
					R{1.25cm}R{1.25cm}R{1.25cm}R{1.25cm}
					@{\hspace{1.2em}}
					R{1.25cm}R{1.25cm}R{1.25cm}R{1.25cm}
					@{\hspace{1.2em}}
					R{1.15cm}@{}}
				\toprule
				& \multicolumn{4}{c}{ Binomial tree by Li \cite{li2000new} (RPE)}
				& \multicolumn{4}{c}{ Binomial tree by Moriggia et al. \cite{moriggia2009} (RPE)}
				& { PDE} \tabularnewline
				
				{$\hfill N_T=$}
				& {$45$} & {$90$} & {$180$} & {$360$}
				& {$45$} & {$90$} & {$180$} & {$360$}
				&{ (price)} \tabularnewline
				
				\cmidrule(lr){2-5}\cmidrule(lr){6-9}\cmidrule(lr){10-10}
				{$K=500$}
				& $\underset{\left(-2.4\%\right)}{\nd{527.6305}}$
				& $\underset{\left(-2.7\%\right)}{\nd{526.2066}}$
				& $\underset{\left(-2.9\%\right)}{\nd{524.8996}}$
				& $\underset{\left(-3.0\%\right)}{\nd{524.1438}}$
				
				& $\underset{\left(-2.3\%\right)}{\nd{528.1826}}$
				& $\underset{\left(-1.8\%\right)}{\nd{530.7757}}$
				& $\underset{\left(-2.3\%\right)}{\nd{528.1184}}$
				& $\underset{\left(-3.9\%\right)}{\nd{519.6053}}$
				& {$540.54$}\tabularnewline
				
				{$K=750$}
				& $\underset{\left(-9.6\%\right)}{\nd{335.6443}}$
				& $\underset{\left(-11.1\%\right)}{\nd{330.1010}}$
				& $\underset{\left(-12.3\%\right)}{\nd{325.6979}}$
				& $\underset{\left(-13.1\%\right)}{\nd{322.8300}}$
				
				& $\underset{\left(-7.7\%\right)}{\nd{342.6237}}$
				& $\underset{\left(-1.8\%\right)}{\nd{364.4434}}$
				& $\underset{\left(-6.4\%\right)}{\nd{347.5553}}$
				& $\underset{\left(-23.9\%\right)}{\nd{282.4093}}$
				& {$371.29$}\tabularnewline
				
				{$K=1000$}
				& $\underset{\left(-18.8\%\right)}{\nd{205.6144}}$
				& $\underset{\left(-22.4\%\right)}{\nd{196.4141}}$
				& $\underset{\left(-24.7\%\right)}{\nd{190.5151}}$
				& $\underset{\left(-26.7\%\right)}{\nd{185.6084}}$
				
				& $\underset{\left(-1.8\%\right)}{\nd{248.4572}}$
				& $\underset{\left(-0.9\%\right)}{\nd{250.8123}}$
				& $\underset{\left(-3.0\%\right)}{\nd{245.5158}}$
				& $\underset{\left(-34.9\%\right)}{\nd{164.8701}}$
				& {$253.07$}\tabularnewline
				
				{$K=1250$}
				& $\underset{\left(-27.7\%\right)}{\nd{125.1529}}$
				& $\underset{\left(-32.2\%\right)}{\nd{117.3217}}$
				& $\underset{\left(-36.0\%\right)}{\nd{110.7800}}$
				& $\underset{\left(-39.1\%\right)}{\nd{105.3589}}$
				
				& $\underset{\left(0.4\%\right)}{\nd{173.8643}}$
				& $\underset{\left(-0.2\%\right)}{\nd{172.7584}}$
				& $\underset{\left(-3.7\%\right)}{\nd{166.7259}}$
				& $\underset{\left(-45.9\%\right)}{\nd{93.7239}}$
				& {$173.12$}\tabularnewline
				
				{$K=1500$}
				 & $\underset{\left(-32.7\%\right)}{\nd{80.3190}}$
				& $\underset{\left(-39.2\%\right)}{\nd{72.6441}}$
				& $\underset{\left(-44.5\%\right)}{\nd{66.2367}}$
				& $\underset{\left(-48.8\%\right)}{\nd{61.1536}}$
				
				& $\underset{\left(-0.4\%\right)}{\nd{118.9447}}$
				& $\underset{\left(-0.1\%\right)}{\nd{119.2611}}$
				& $\underset{\left(-15.4\%\right)}{\nd{101.0115}}$
				& $\underset{\left(-68.2\%\right)}{\nd{37.9191}}$
				& {$119.39$}\tabularnewline
				\bottomrule
			\end{tabular*}
		}{\par}
		\par\end{centering}
	\caption{\label{tab:Li_Moriggia} Pricing results for the benchmark binomial trees of Li and of Moriggia et al.\ for different strikes and tree sizes. The entries report the prices and, in parentheses, the relative percentage  errors (RPE) with respect to the PDE benchmark price. In both cases, the tree is constructed from the exact local-volatility function.}

\vspace{0.8\baselineskip}
	\begin{centering}
			\small{	\begin{tabular*}{\textwidth}{@{\extracolsep{\fill}}l
			R{1.25cm}R{1.25cm}R{1.25cm}R{1.25cm}
			@{\hspace{1.2em}}
			R{1.25cm}R{1.25cm}R{1.25cm}R{1.25cm}
			@{\hspace{1.2em}}
			R{1.15cm}@{}}
		\toprule
		&
		\multicolumn{4}{c}{ Training grid $10\times 20$ (RPE)}
		&
		\multicolumn{4}{c}{ Training grid $3\times 6$ (RPE)}
		&
		PDE
		\tabularnewline
		
		{$\hfill N_T=$}
		& {$45$} & {$90$} & {$180$} & {$360$}
		& {$45$} & {$90$} & {$180$} & {$360$}
		& {(price)}
		\tabularnewline
		
		\cmidrule(lr){2-5}\cmidrule(lr){6-9}\cmidrule(lr){10-10}
		{$K=500$}
 & $\underset{\left(0.1\%\right)}{\nd{541.2939}}$
& $\underset{\left(0.0\%\right)}{\nd{540.6756}}$
& $\underset{\left(0.0\%\right)}{\nd{540.5831}}$
& $\underset{\left(-0.0\%\right)}{\nd{540.5359}}$
	
 & $\underset{\left(-0.1\%\right)}{\nd{539.9274}}$
& $\underset{\left(-0.0\%\right)}{\nd{540.4579}}$
& $\underset{\left(-0.0\%\right)}{\nd{540.4109}}$
& $\underset{\left(-0.0\%\right)}{\nd{540.4982}}$
		
		& {$540.54$}
		\tabularnewline
		
		{$K=750$}
 & $\underset{\left(-0.8\%\right)}{\nd{368.2319}}$
& $\underset{\left(0.0\%\right)}{\nd{371.3390}}$
& $\underset{\left(-0.0\%\right)}{\nd{371.1721}}$
& $\underset{\left(0.0\%\right)}{\nd{371.3587}}$
		
 & $\underset{\left(0.8\%\right)}{\nd{374.0947}}$
& $\underset{\left(0.3\%\right)}{\nd{372.4149}}$
& $\underset{\left(-0.3\%\right)}{\nd{370.1848}}$
& $\underset{\left(-0.0\%\right)}{\nd{371.1462}}$
		& {$371.29$}
		\tabularnewline
		
			{$K=1000$}
 & $\underset{\left(-0.2\%\right)}{\nd{252.6711}}$
& $\underset{\left(0.2\%\right)}{\nd{253.6547}}$
& $\underset{\left(0.2\%\right)}{\nd{253.4737}}$
& $\underset{\left(-0.0\%\right)}{\nd{252.9792}}$

& $\underset{\left(0.8\%\right)}{\nd{255.0683}}$
& $\underset{\left(-0.3\%\right)}{\nd{252.2769}}$
& $\underset{\left(-0.4\%\right)}{\nd{252.0827}}$
& $\underset{\left(-0.2\%\right)}{\nd{252.5243}}$
		& {$253.07$}
		\tabularnewline
		
		{$K=1250$}
 & $\underset{\left(-1.0\%\right)}{\nd{171.3807}}$
& $\underset{\left(-0.7\%\right)}{\nd{171.9215}}$
& $\underset{\left(0.2\%\right)}{\nd{173.4355}}$
& $\underset{\left(0.1\%\right)}{\nd{173.2792}}$

& $\underset{\left(0.4\%\right)}{\nd{173.7427}}$
& $\underset{\left(-0.4\%\right)}{\nd{172.4238}}$
& $\underset{\left(-0.5\%\right)}{\nd{172.1982}}$
& $\underset{\left(-0.3\%\right)}{\nd{172.6150}}$
		& {$173.12$}
		\tabularnewline
		
		{$K=1500$}
 & $\underset{\left(-0.9\%\right)}{\nd{118.2755}}$
& $\underset{\left(-1.1\%\right)}{\nd{118.0222}}$
& $\underset{\left(-0.7\%\right)}{\nd{118.6086}}$
& $\underset{\left(0.4\%\right)}{\nd{119.8624}}$
		
 & $\underset{\left(1.2\%\right)}{\nd{120.7765}}$
& $\underset{\left(-1.0\%\right)}{\nd{118.2434}}$
& $\underset{\left(-0.9\%\right)}{\nd{118.3654}}$
& $\underset{\left(-0.8\%\right)}{\nd{118.3828}}$
		& {$119.39$}
		\tabularnewline
		\bottomrule
	\end{tabular*}}
			\par\end{centering}
		\caption{\label{tab:Price_convergence}Pricing results for European call prices produced by the proposed method for different strikes and tree sizes. The entries report the prices and, in parentheses, the relative percentage errors (RPE) with respect to the PDE benchmark price. Results are shown for the two trained models corresponding to the $10\times 20$ and $3\times 6$ training grids. The cumulative number of training epochs is $N_E=8000$ and the penalization coefficient is $\lambda_{\mathrm{space}}=0$.}
	\end{table}

	The results in \Cref{tab:Price_convergence} show a clear improvement over the two other approaches discussed above. For both training configurations, the pricing errors remain small across all reported strikes and tree sizes, and no severe deterioration is observed as $N_T$ increases. The denser $10\times 20$ training grid generally provides the most accurate results, especially for the finer trees, where the relative percentage errors are very close to zero for most strikes. For the coarser $3\times 6$ training grid, the errors are slightly less uniform, with deviations of about $1\%$ in some cases, but the overall behavior remains stable and satisfactory.
	The relative percentage error tends to become larger in magnitude as the benchmark option price decreases, since the training loss is constructed from absolute pricing discrepancies rather than relative ones. This is ultimately a choice in the calibration objective and could be modified by introducing a relative-error weighting in the loss function.

    \Cref{fig:tree} compares, for the case $N_T=25$, the initial tree and the final calibrated tree, and also shows the corresponding log-deformation heat map.  Panel (a) shows that the calibration  produces a structured reshaping of the lattice across both time and state dimensions.
    Panel (b) confirms that the learned log-deformation remains relatively small around the central region of the tree, while it becomes positive in the upper branch and negative in the lower branch, with increasing magnitude as time progresses.
	
	\begin{figure}[t]
		\centering
		\includegraphics[width=\textwidth]{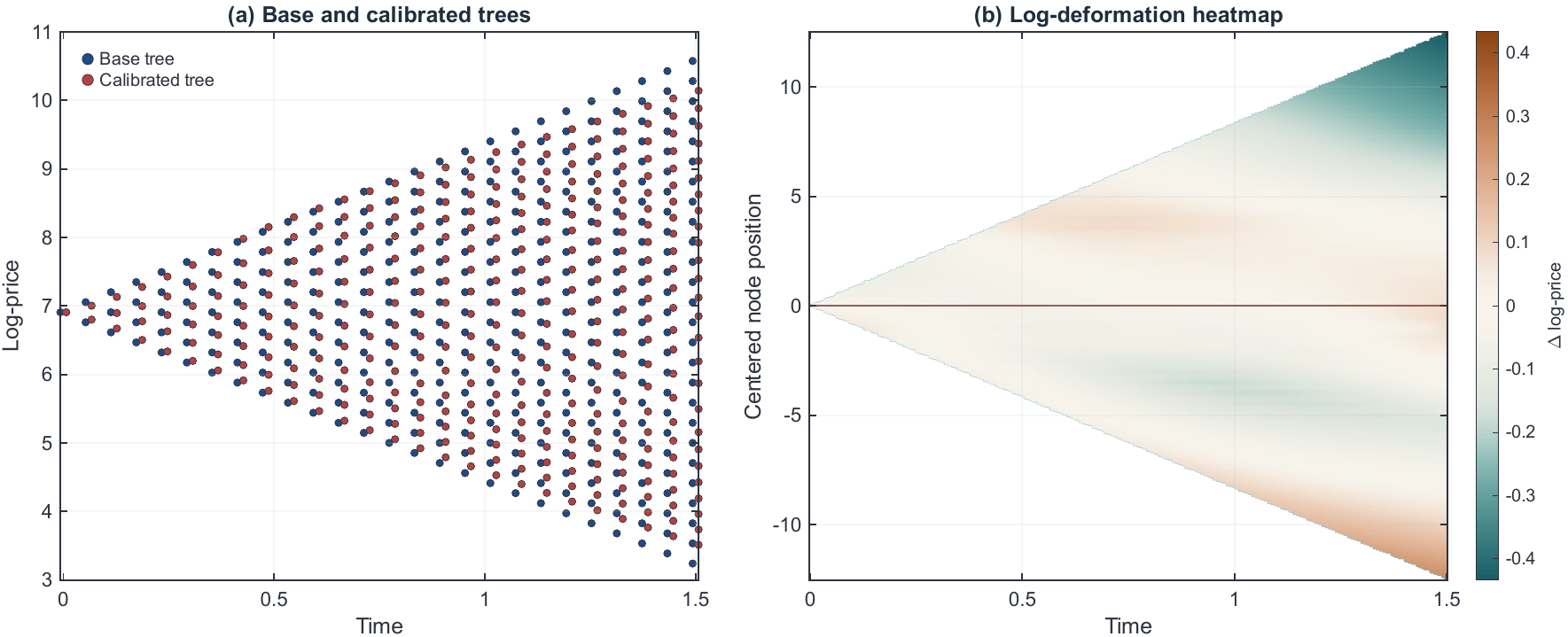}
		\caption{Base and calibrated recombining trees for $N_T=25$. Left: node locations in log-price coordinates for the benchmark and calibrated trees; a small horizontal offset is introduced for visual clarity. Right: heatmap of the log-deformation $\log S^{(\theta)}_{n,j}-\log S^{0}_{n,j}$, centered vertically for readability.}
		\label{fig:tree}
	\end{figure}

	\begin{table}
		\begin{centering}
	\small{\begin{tabular*}{\textwidth}{@{\extracolsep{\fill}}llcccccccccccccccc@{}}
			\toprule
			$\hfill N_{T}=$ 
			& & \multicolumn{4}{c}{$180$}
			& & \multicolumn{4}{c}{$360$}
			& & \multicolumn{4}{c}{$720$}
			& & \tabularnewline
			
			$\hfill N_{E}=$ 
			& & $ {1000}$ & $2000$ & $4000$ & $8000$
			& & $1000$ & $ {2000}$ & $4000$ & $8000$
			& & $1000$ & $ {2000}$ & $ {4000}$ & $8000$
			& & DSCL \tabularnewline
			
			\cmidrule{3-6}\cmidrule{8-11}\cmidrule{13-16}\cmidrule{18-18}
			
			\multicolumn{18}{l}{\rule{0pt}{5mm}Input grid $N_{T}^{opt}\times N_{K}^{opt}=10\times20$}
			\tabularnewline
			
			\midrule
			
			OP IS ($\times10^{-2}$) 
			& & \nz{25.56}& \nz{18.91}& \nz{16.82}& \nz{15.38}
			& & \nz{18.38}& \nz{10.55}& \nz{7.42}& \nz{5.86}
			& & \nz{14.56}& \nz{8.39}& \nz{5.83}& \nz{4.44}
			& & $94$ \tabularnewline
			
			OP OOS ($\times10^{-2}$) 
			& & \nz{26.19}& \nz{23.02}& \nz{22.41}& \nz{22.27}
			& & \nz{17.41}& \nz{12.35}& \nz{11.81}& \nz{11.97}
			& & \nz{12.57}& \nz{8.00}& \nz{6.77}& \nz{6.41}
			& & $108$ \tabularnewline

				RT ($\times10^3$) 
			& & \nl{75}& \nl{124}& \nl{205}& \nl{346}
			& & \nl{206}& \nl{363}& \nl{627}& \nl{1131}
			& & \nl{666}& \nl{1149}& \nl{1959}& \nl{3582}
			& & \nl{9500} \tabularnewline
			
			\multicolumn{7}{l}{\rule{0pt}{5mm}Input grid $N_{T}^{opt}\times N_{K}^{opt}=3\times6$}
			& & & & & & & & & & \tabularnewline
			
			\midrule
			
			OP IS ($\times10^{-2}$) 
			& & \nz{8.95}& \nz{3.07}& \nz{1.60}& \nz{0.66}
			& & \nz{14.78}& \nz{5.54}& \nz{2.46}& \nz{0.64}
			& & \nz{17.76}& \nz{6.73}& \nz{6.73}& \nz{1.81}
			& & $162$ \tabularnewline
			
			OP OOS ($\times10^{-2}$) 
			& & \nz{186.97}& \nz{190.38}& \nz{191.03}& \nz{195.15}
			& & \nz{128.31}& \nz{126.68}& \nz{132.88}& \nz{136.93}
			& & \nz{128.01}& \nz{126.83}& \nz{126.83}& \nz{121.59}
			& & $181$ \tabularnewline

				RT ($\times10^3$) 
			& & \nl{77}& \nl{126}& \nl{205}& \nl{358}
			& & \nl{210}& \nl{358}& \nl{610}& \nl{1096}
			& & \nl{681}& \nl{1159}& \nl{1960}& \nl{3571}
			& & \nl{9500} \tabularnewline
			
				\multicolumn{7}{l}{\rule{0pt}{5mm}Input grid $N_{T}^{opt}\times N_{K}^{opt}=256\times256$}
			& & & & & & & & & & \tabularnewline
			
			\midrule
			
			OP IS ($\times10^{-2}$) 
			& & \nz{26.26}& \nz{18.41}& \nz{16.77}& \nz{16.13}		
			& & \nz{23.62}& \nz{13.81}& \nz{9.88}& \nz{8.61}			
			& & \nz{26.18}& \nz{11.43}& \nz{6.55}& \nz{4.90}
			&  $ $ \tabularnewline

			RT  ($\times10^3$) 
			& & \nl{513}& \nl{835}& \nl{1397}& \nl{2431}
			& & \nl{980}& \nl{1780}& \nl{3262}& \nl{6023}
			& & \nl{1885}& \nl{3710}& \nl{7109}& \nl{13448}
			&   \tabularnewline
			
			\bottomrule
		\end{tabular*}}

			\par\end{centering}
			\caption{\label{tab:RMSE_1}RMSE for option pricing (OP) in-sample (IS) and out-of-sample (OOS), expressed in units of \(10^{-2}\). The values of \(N_T\) are chosen so that the option maturities used in training coincide with tree layers, and no interpolation in time is required. For the \(256\times256\) input grid, only the in-sample error is reported, since the calibration grid coincides with the dense evaluation grid. Runtime (RT) is reported in units of \(10^3\) seconds. The last column reports
				the DSCL benchmark from Wang et al.~\cite{wang2025}; the corresponding runtime
				should be interpreted as indicative, since implementations and hardware
				optimization differ. }
	\end{table}
	
 \medskip

We next turn to  a comparison with the   results reported by \citet{wang2025}. The results are summarized in \Cref{tab:RMSE_1}, where the number of tree time steps is $N_T\in\{180,360,720\}$.
 We report option-pricing RMSEs for our trees both in-sample (OP IS) and out-of-sample (OP OOS),  and the total runtime (RT), for increasing cumulative training budgets $N_E$ and for different numbers of tree time steps $N_T$. As in the previous experiments, the spatial regularization is switched off, that is, $\lambda_{\mathrm{space}}=0$. The last column (DSCL) shows the values reported in \citet{wang2025}.

The in-sample pricing errors generally decrease as the cumulative training
budget \(N_E\) increases, especially on the dense \(10\times20\) grid. The
out-of-sample behavior is more nuanced: it improves clearly in the dense case,
whereas on the sparse \(3\times6\) grid it is not monotone in \(N_E\), reflecting
the weaker identifiability of the calibration problem under sparse training
data. Increasing the number of time steps tends to improve pricing accuracy, especially for the denser $10\times 20$ training grid, where RMSEs for both OP IS and OP OOS  decrease substantially as $N_T$ grows. The comparison with the DSCL benchmark is favourable in this experimental setup, especially in
out-of-sample pricing errors, although runtime figures should be interpreted as indicative because the
implementations and hardware optimizations differ.

Results in the last two rows correspond to  training on
a very dense set of option prices, 
 so the resulting in-sample errors 
may be interpreted as an approximate lower bound on the discrepancies attainable by the model, for a fixed value of $N_T$. It shows that errors should not be attributed solely to the sparsity of the training set, but also to the approximation limits of the tree discretization itself.
At the same time, the comparison with the standard $10\times20$ training grid is encouraging:
once the tree is sufficiently refined, a moderately dense set of option quotes already captures most of the information needed to identify an accurate deformation of the base lattice.

	\FloatBarrier
	\subsubsection{Local volatility surfaces}
	
	We next turn to the local-volatility structure implied
	by the calibrated tree. The quantities
	\({\treevol_{n,j}}^{(\theta)}\)  provide useful
	information, and in the
	synthetic setting, where the local volatility \(\locvol(t,S)\) is
	known, they make it possible to assess  whether the tree-implied
	local-volatility profile is consistent with the data-generating model.

	Table~\ref{tab:LV_1} shows the local-volatility diagnostics for the construction described in \Cref{sec:gpr-local-vol}. The GPR-based errors (\emph{LV N-GPR} and \emph{LV G-GPR}) are uniformly smaller than the corresponding raw errors (\emph{LV N-Raw}), indicating that the local-volatility values extracted directly from the tree retain some irregularities that are effectively smoothed out by the regression step. The final results are particularly satisfactory in the dense \(10\times20\) setting, where, for sufficiently fine trees, the reconstructed local-volatility errors become comparable with, and in some cases slightly improve upon, the benchmark reported by Wang et al. The \(3\times6\) case remains more demanding, which is consistent with the weaker identifiability of the local-volatility structure.
	
		\begin{table}
		\begin{centering}
			\small{\begin{tabular*}{\textwidth}{@{\extracolsep{\fill}}llcccccccccccccccc@{}}
					\toprule
					$\hfill N_{T}=$ 
					& & \multicolumn{4}{c}{$180$}
					& & \multicolumn{4}{c}{$360$}
					& & \multicolumn{4}{c}{$720$}
					& & \tabularnewline
					
					$\hfill N_{E}=$ 
					& & $1000$ & $2000$ & $4000$ & $8000$
					& & $1000$ & $2000$ & $4000$ & $8000$
					& & $1000$ & $2000$ & $4000$ & $8000$
					& & DSCL \tabularnewline
					
					\cmidrule{3-6}\cmidrule{8-11}\cmidrule{13-16}\cmidrule{18-18}
					
					\multicolumn{18}{l}{\rule{0pt}{5mm}Input grid $N_{T}^{opt}\times N_{K}^{opt}=10\times20$}
				  \tabularnewline
					
					\midrule
					
					LV N-Raw ($\%$)
					& & \no{1.789}& \no{1.549}& \no{1.477}& \no{1.520}
					& & \no{1.836}& \no{1.641}& \no{1.694}& \no{1.743}
					& & \no{1.591}& \no{1.397}& \no{1.380}& \no{1.328}
					& &  \tabularnewline
					
					LV N-GPR ($\%$)
					& & \no{0.938}& \no{0.774}& \no{0.699}& \no{0.656}
					& & \no{1.050}& \no{0.808}& \no{0.680}& \no{0.620}
					& & \no{0.793}& \no{0.545}& \no{0.435}& \no{0.379}
					& &  \tabularnewline

					LV G-GPR ($\%$)
					& & \no{1.095}& \no{0.839}& \no{0.696}& \no{0.625}
					& & \no{1.205}& \no{0.948}& \no{0.791}& \no{0.698}
					& & \no{0.826}& \no{0.586}& \no{0.466}& \no{0.397}
					& & $1$ \tabularnewline

					\multicolumn{18}{l}{\rule{0pt}{5mm}Input grid $N_{T}^{opt}\times N_{K}^{opt}=3\times6$}
				 \tabularnewline
					
					\midrule
					
					LV N-Raw ($\%$)
					& & \no{5.038}& \no{5.160}& \no{5.179}& \no{5.257}
					& & \no{4.139}& \no{3.966}& \no{4.108}& \no{4.275}
					& & \no{4.150}& \no{4.152}& \no{4.152}& \no{4.058}
					& &   \tabularnewline
					
					LV N-GPR ($\%$)
					& & \no{3.977}& \no{4.086}& \no{4.093}& \no{4.096}
					& & \no{3.173}& \no{3.027}& \no{3.071}& \no{3.066}
					& & \no{3.169}& \no{3.172}& \no{3.172}& \no{2.954}
					& &   \tabularnewline

					LV G-GPR ($\%$)
					& & \no{4.618}& \no{4.719}& \no{4.723}& \no{4.762}
					& & \no{3.723}& \no{3.661}& \no{3.758}& \no{3.784}
					& & \no{3.572}& \no{3.704}& \no{3.704}& \no{3.490}
					& & $2$ \tabularnewline
					
					\multicolumn{18}{l}{\rule{0pt}{5mm}Input grid $N_{T}^{opt}\times N_{K}^{opt}=256\times256$}
					 \tabularnewline
					
					\midrule
					
					LV N-Raw ($\%$)
					& & \no{2.017}& \no{1.738}& \no{1.628}& \no{1.392}
					& & \no{2.198}& \no{1.963}& \no{1.554}& \no{1.272}
					& & \no{2.568}& \no{1.670}& \no{1.377}& \no{1.218}
					& &   \tabularnewline
					
					LV N-GPR  ($\%$)
					& & \no{1.220}& \no{0.837}& \no{0.736}& \no{0.717}
					& & \no{1.559}& \no{1.115}& \no{0.824}& \no{0.622}
					& & \no{1.645}& \no{0.907}& \no{0.699}& \no{0.603}
					& &   \tabularnewline

					LV G-GPR ($\%$)
					& &\no{1.379}& \no{0.935}& \no{0.838}& \no{0.811}
					& & \no{1.812}& \no{1.161}& \no{0.877}& \no{0.688}
					& & \no{1.605}& \no{0.857}& \no{0.723}& \no{0.658}
					& &  \tabularnewline
					
					\bottomrule
			\end{tabular*}}
			\par\end{centering}
		\caption{\label{tab:LV_1}
			Relative RMSEs, expressed in percent, for the local-volatility diagnostics in the synthetic experiments, reported for different tree sizes \(N_T\) and cumulative training budgets \(N_E\), and for the two input grids \(10\times20\) and \(3\times6\). \emph{LV N-Raw} is the  relative RMSE computed directly from the tree-implied local volatility, \emph{LV N-GPR} is the relative RMSE after GPR reconstruction, and \emph{LV G-GPR} is the relative RMSE on the dense \(256\times256\) validation grid after GPR reconstruction. The last column reports the DSCL benchmark from \citet{wang2025}.}
	\end{table}

\Cref{fig:lv} compares the GPR-reconstructed tree-implied local-volatility surfaces. In the dense \(10\times20\) setting, the reconstructed surface remains close to the analytical benchmark over most of the core region of the \((T,S/S_0)\) domain. In the sparse \(3\times6\) setting, the GPR-based reconstruction is naturally less stable, but it still provides a useful qualitative diagnostic of the deformation induced by calibration. As expected, a richer set of option quotes not only improves repricing accuracy, but also stabilizes the reconstruction of the local-volatility profile.

	\begin{figure}[p] 
	\centering	
	\includegraphics[width=0.7\textwidth]{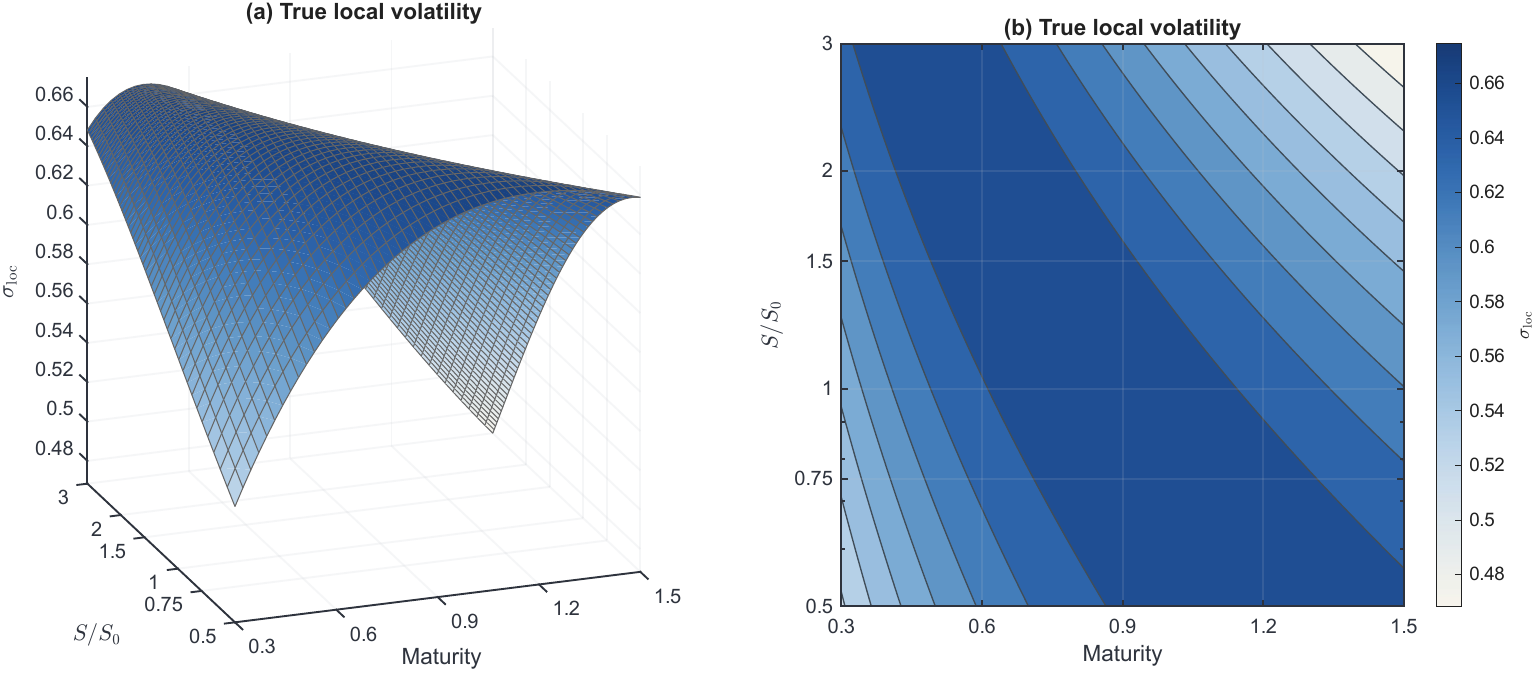}
\caption{\label{fig:lv_benchmark}
	Benchmark local-volatility surface in the synthetic setting, displayed over the common validation domain in maturity and spot ratio \(S/S_0\).  }
\end{figure}

	\begin{figure}[p]
		\centering
		\begin{subfigure}[b]{\textwidth}
			\centering
			\includegraphics[width=\textwidth]{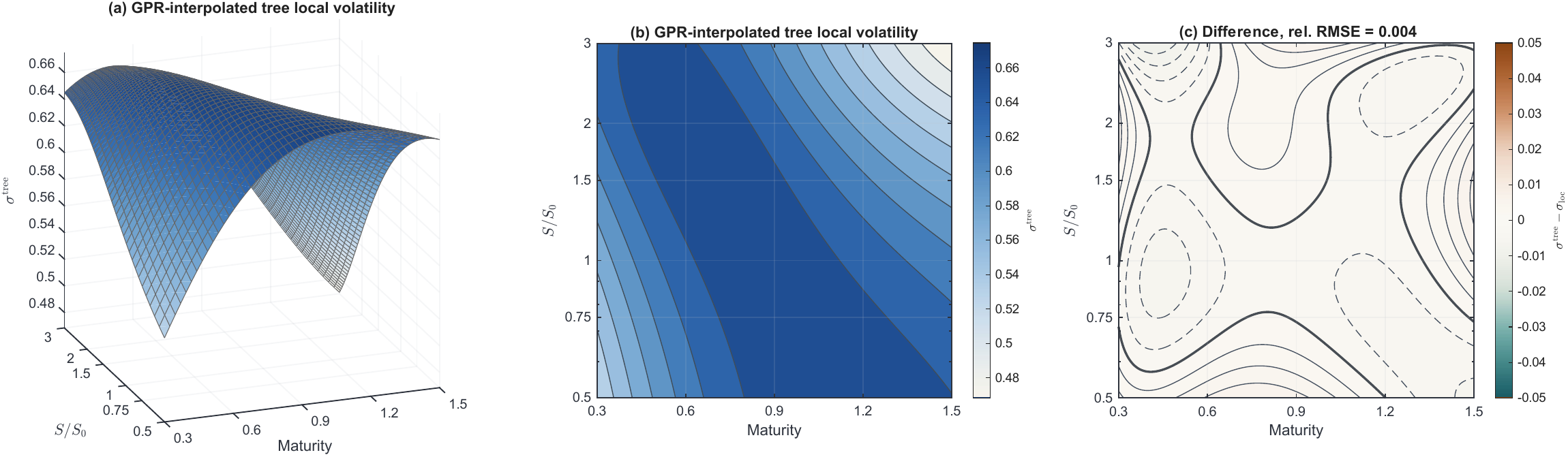}
			\caption{10x20 training grid.}
		\end{subfigure}
		
		\vspace{0.2em}
		
		\begin{subfigure}[b]{\textwidth}
			\centering
			\includegraphics[width=\textwidth]{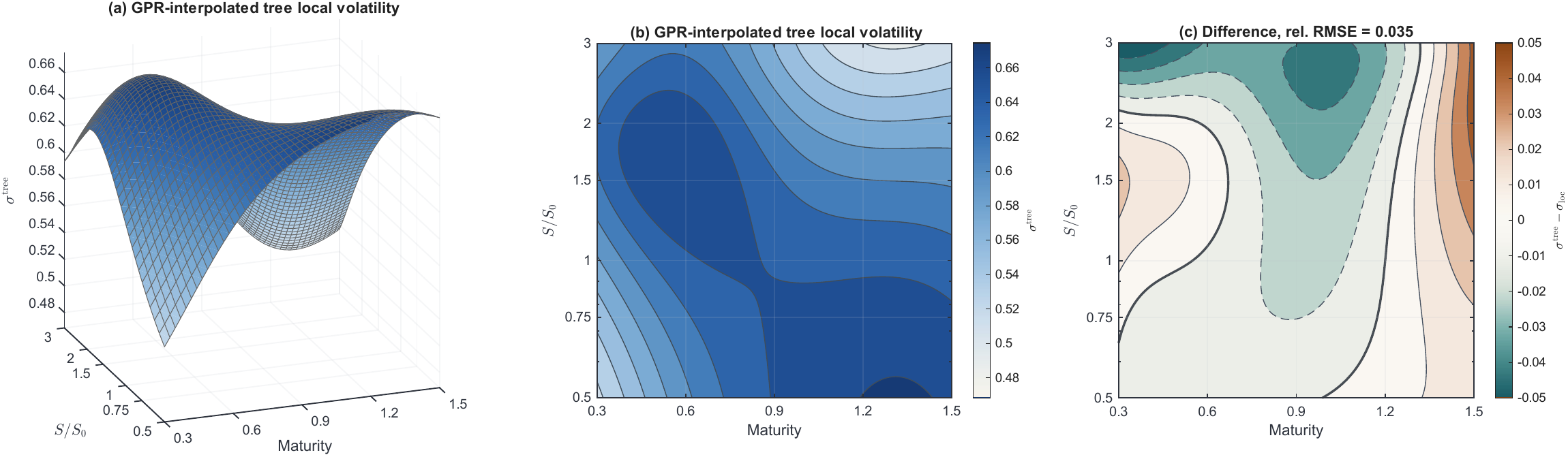}
			\caption{3x6 training grid.}
		\end{subfigure}
		
\caption{Comparison between the true local volatility surface and the
	GPR-reconstructed tree-implied local volatility surface. For each grid,
	panel (a) shows the GPR-reconstructed tree-implied local volatility surface,
	panel (b) displays the same surface as a contour plot, and panel (c) reports
	the difference \(\treevol-\locvol\). In panel
	(c), dashed contour lines correspond to negative levels, solid contour lines
	to positive levels, and the thicker solid contour denotes the zero level.
	Here \(N_T=720\), \(N_E=8000\), and \(\lambda_{\mathrm{space}}=0\).}
		\label{fig:lv}
	\end{figure}

	%\FloatBarrier
	
	\subsubsection{Impact of spatial regularization}
	
	We finally examine the role of the spatial regularization coefficient by fixing
	the cumulative training budget at the common checkpoint \(N_E=8000\) and taking
	$
	\lambda_{\mathrm{space}} \in \{0,3,10,30,100\}.
	$
	The value
	\(N_E=8000\) is used here to compare different regularization strengths at the
	same late-stage checkpoint, and should not be interpreted as a necessary
	training horizon. The purpose of this experiment is to assess whether the
	penalty term can improve the
	stability of calibration and the quality of the implied local-volatility
	diagnostics.
	
	As shown in Table~\ref{tab:lambda_pricing}, the effect of spatial
	regularization on pricing accuracy is not monotone. For the dense
	\(10\times20\) training grid, the gains are modest: regularization improves the
	out-of-sample RMSE for some tree sizes, but the unregularized calibration is
	already very competitive, especially for the finest tree. By contrast, the
	effect is much more visible for the sparse \(3\times6\) grid. In this case,
	introducing a positive value of \(\lambda_{\mathrm{space}}\) substantially
	reduces the out-of-sample pricing error for all reported tree sizes. 
	
	The non-monotonicity is consistent with the
	sensitivity analysis in \citet{wang2025}: a regularization weight that is too
	small may leave the fit essentially unregularized, whereas a weight that is too
	large may slow optimization and worsen calibration accuracy for a fixed training
	budget.
	Table~\ref{tab:lambda_lv} shows a similar pattern for the local-volatility
	diagnostics:  the penalty is most useful when the
	inverse problem is weakly identified.
     Hence \(\lambda_{\mathrm{space}}\) should not be increased mechanically, but determined after checking results for a  small grid of
	candidate values and choosing the weakest regularization that stabilizes the tree.
	
		\begin{table}
		\begin{centering}
			\small{
\begin{tabular*}{\textwidth}{@{\extracolsep{\fill}}llccccccccccccccccccc@{}}
	\toprule
	$\hfill N_{T}=$
	& & \multicolumn{5}{c}{$180$}
	& & \multicolumn{5}{c}{$360$}
	& & \multicolumn{5}{c}{$720$}
	& & \tabularnewline
	
	$\hfill \lambda_{\mathrm{space}}=$
	& & $0$ & $3$ & $10$ & $30$ & $100$
	& & $0$ & $3$ & $10$ & $30$ & $100$
	& & $0$ & $3$ & $10$ & $30$ & $100$
	& & DSCL \tabularnewline
	
	\cmidrule{3-7}\cmidrule{9-13}\cmidrule{15-19}\cmidrule{21-21}
	
	\multicolumn{21}{l}{\rule{0pt}{5mm}Input grid $N_{T}^{opt}\times N_{K}^{opt}=10\times20$}
	\tabularnewline
	
	\midrule
	
	OP IS ($\times10^{-2}$) 
	& & \nz{15.38}& \nz{15.95}& \nz{17.85}& \gc\nz{18.60}& \nz{24.78}
	& & \nz{5.86}& \gc\nz{8.04}& \nz{8.45}& \nz{12.51}& \nz{16.60}
	& & \gc\nz{4.44}& \nz{4.89}& \nz{6.64}& \nz{9.18}& \nz{14.25}
	& & $94$ \tabularnewline
	
	OP OOS ($\times10^{-2}$) 
	& & \nz{22.27}& \nz{21.65}& \nz{21.63}&  \gc\nz{20.69}& \nz{25.17}
	& & \nz{11.97}& \gc\nz{11.22} & \nz{11.36}& \nz{12.48}& \nz{15.99}
	& & \gc\nz{6.41}& \nz{6.63}& \nz{6.58}& \nz{8.02}& \nz{11.89}
	& & $108$ \tabularnewline
	
	RT ($\times10^3$)  
	& & \nl{324}& \nl{404}& \nl{437}& \gc\nl{468}& \nl{498}
	& & \nl{1198}& \gc\nl{1076}& \nl{1157}& \nl{1193}& \nl{1239}
	& & \gc\nl{3590}& \nl{3655}& \nl{3744}& \nl{3806}& \nl{3904}
	& & \nl{9500} \tabularnewline
	
	\multicolumn{21}{l}{\rule{0pt}{5mm}Input grid $N_{T}^{opt}\times N_{K}^{opt}=3\times6$}
	\tabularnewline
	
	\midrule
	
	OP IS ($\times10^{-2}$) 
	& & \nz{0.66}& \nz{2.26}& \nz{2.79}& \nz{4.66}& \gc\nz{9.85}
	& & \nz{0.64}& \nz{2.57}& \gc\nz{3.83}& \nz{4.65}& \nz{7.81}
	& & \nz{1.81}& \gc\nz{3.87}& \nz{3.16}& \nz{4.84}& \nz{8.20}
	& & $162$ \tabularnewline
	
	OP OOS  ($\times10^{-2}$) 
	& & \nz{195.15}& \nz{117.24}& \nz{105.75}& \nz{105.31}&  \gc\nz{100.97}
	& & \nz{136.93}& \nz{100.35}& \gc\nz{86.66}& \nz{87.70}& \nz{94.43}
	& & \nz{121.59}& \gc\nz{75.80}& \nz{88.17}& \nz{81.32}& \nz{89.55}
	& & $181$ \tabularnewline
	
	RT ($\times10^3$) 
	& & \nl{298}& \nl{387}& \nl{425}& \nl{447}& \gc\nl{482}
	& & \nl{1066}& \nl{1043}& \gc\nl{1114}& \nl{1147}& \nl{1183}
	& & \nl{3647}& \gc\nl{3605}& \nl{3687}& \nl{3730}& \nl{3831}
	& & \nl{9500} \tabularnewline
	
	\bottomrule 
\end{tabular*}

				}
			\par\end{centering} 
	\caption{RMSE for option pricing (OP) in-sample (IS) and out-of-sample (OOS),
		expressed in units of $10^{-2}$, together with runtime (RT), for different
		values of the spatial regularization coefficient
        $\lambda_{\mathrm{space}}$,
		with cumulative training budget fixed at 
        $N_E=8000$. Runtime is expressed in
		units of $10^3$ seconds. Grey-shaded entries identify, for each input grid and
		each fixed value of $N_T$, the value of $\lambda_{\mathrm{space}}$ that
		minimizes the OP OOS error; the corresponding OP IS and RT values are shaded
		as well.  } 
	
	\label{tab:lambda_pricing}
	\end{table}

	\begin{table}
		\begin{centering}
			
			\small{\begin{tabular*}{\textwidth}{@{\extracolsep{\fill}}llccccccccccccccccccc@{}}
				\toprule
				$\hfill N_{T}=$
				& & \multicolumn{5}{c}{$180$}
				& & \multicolumn{5}{c}{$360$}
				& & \multicolumn{5}{c}{$720$}
				& & \tabularnewline
				
				$\hfill \lambda_{\mathrm{space}}=$
				& & $0$ & $3$ & $10$ & $30$ & $100$
				& & $0$ & $3$ & $10$ & $30$ & $100$
				& & $0$ & $3$ & $10$ & $30$ & $100$
				& & DSCL \tabularnewline
				
				\cmidrule{3-7}\cmidrule{9-13}\cmidrule{15-19}\cmidrule{21-21}
				
				\multicolumn{21}{l}{\rule{0pt}{5mm}Input grid $N_{T}^{opt}\times N_{K}^{opt}=10\times20$}
				\tabularnewline
				
				\midrule
				
				LV N-Raw  ($\%$)
				& &\gc \no{1.520}& \no{1.207}& \no{1.195}& \no{1.016}& \no{1.259}
				& &\gc \no{1.743}& \no{1.188}& \no{1.176}& \no{1.136}& \no{1.273}
				& &\gc \no{1.328}& \no{1.152}& \no{1.096}& \no{1.313}& \no{1.381}
				& & \tabularnewline
				
				LV N-GPR ($\%$)
				& &\gc \no{0.656}& \no{0.720}& \no{0.827}& \no{0.830}& \no{1.076}
				& &\gc \no{0.620}& \no{0.707}& \no{0.835}& \no{0.890}& \no{1.004}
				& &\gc \no{0.379}& \no{0.716}& \no{0.762}& \no{1.046}& \no{1.153}
				& &  \tabularnewline
				
				LV G-GPR   ($\%$)
				& &\gc \no{0.625}& \no{0.800}& \no{0.922}& \no{0.846}& \no{1.209}
				& &\gc \no{0.698}& \no{0.896}& \no{1.062}& \no{1.119}& \no{1.239}
				& &\gc \no{0.397}& \no{0.915}& \no{0.967}& \no{1.377}& \no{1.469}
				& & $1$ \tabularnewline
				
				\multicolumn{21}{l}{\rule{0pt}{5mm}Input grid $N_{T}^{opt}\times N_{K}^{opt}=3\times6$}
				\tabularnewline
				
				\midrule
				
				LV N-Raw   ($\%$)
				& &  \no{5.257}& \no{3.498}& \no{3.151}& \no{3.119}&\gc \no{2.994}
				& & \no{4.275}& \no{3.069}& \no{2.958}&\gc \no{2.864}& \no{3.073}
				& & \no{4.058}&\gc \no{2.764}& \no{3.033}& \no{2.901}& \no{2.909}
				& &  \tabularnewline
				
					LV N-GPR  ($\%$)
				& & \no{4.096}& \no{2.898}& \no{2.835}& \no{2.917}&\gc \no{2.931}
				& & \no{3.066}& \no{2.501}& \no{2.505}&\gc \no{2.501}& \no{2.757}
				& & \no{2.954}&\gc \no{2.121}& \no{2.594}& \no{2.505}& \no{2.631}
				& &\tabularnewline
				
					LV G-GPR   ($\%$)
				& & \no{4.762}& \no{3.313}& \no{3.105}& \no{3.140}&\gc \no{3.067}
				& & \no{3.784}& \no{3.015}& \no{2.951}&\gc \no{2.907}& \no{3.001}
				& & \no{3.490}&\gc \no{2.516}& \no{3.124}& \no{3.010}& \no{3.037}
				& &   $2$  \tabularnewline
				
				\bottomrule 
			\end{tabular*}

			}
			\par\end{centering} 
		\caption{Relative RMSEs, expressed in percent, for the local-volatility
			diagnostics in the synthetic experiments, for different values of the spatial
			regularization coefficient $\lambda_{\mathrm{space}}$, with cumulative training
			budget fixed at $N_E=8000$. Grey-shaded entries identify,
			for each input grid and each fixed value of $N_T$, the value of
			$\lambda_{\mathrm{space}}$ that minimizes the LV G-GPR error; the corresponding
			LV N-Raw and LV N-GPR values are shaded as well. } 
		\label{tab:lambda_lv}
	\end{table}
	This is also confirmed in
	Figure~\ref{fig:lambda_lv_surfaces}, which shows the GPR-reconstructed tree-implied
	local-volatility surfaces for the same two representative regularized
	configurations.  Relative to the corresponding unregularized case, there is  no improvement in the dense setting but a better result in the sparse one. 
    
		\begin{figure}[p]
		\centering
		\includegraphics[width=0.95\textwidth]{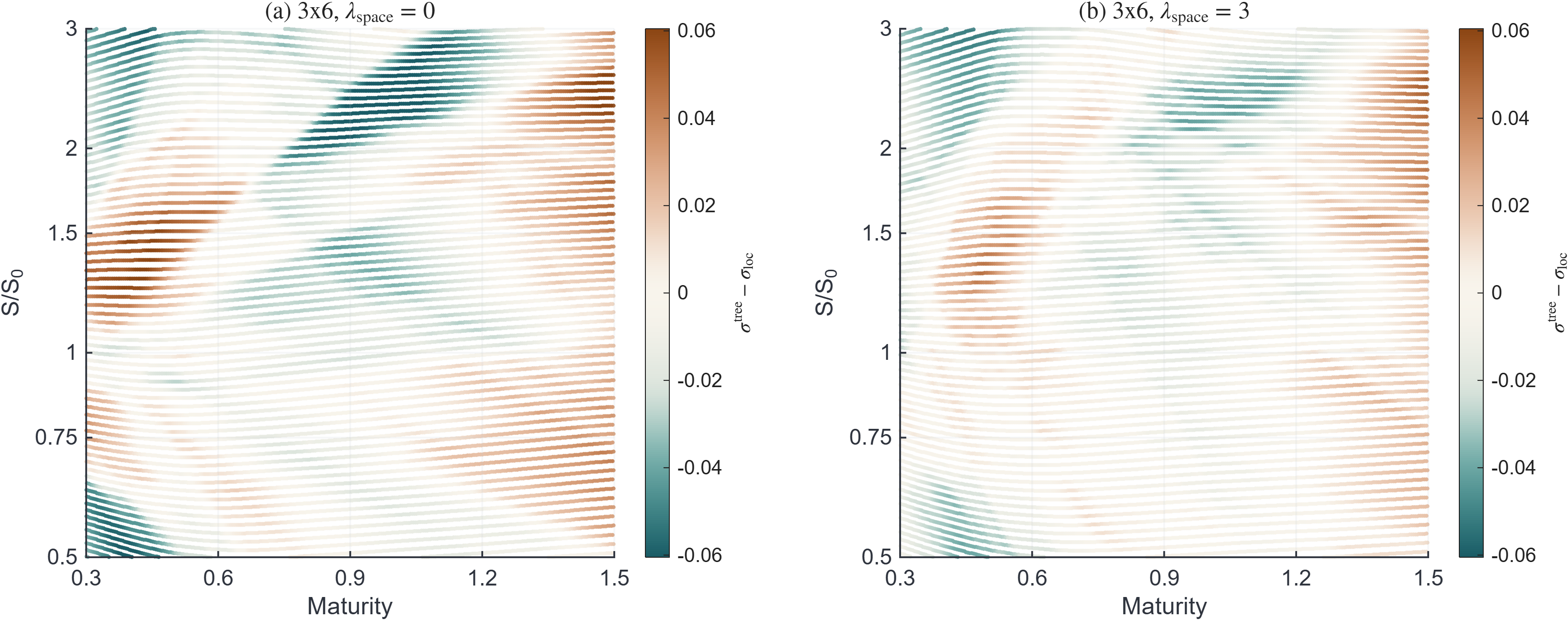}		
		\caption{\label{fig:nodewise-lv-errors}
			Local-volatility error for the $3\times 6$ training grid with $N_T=720$, comparing the unregularized case $\lambda_{\mathrm{space}}=0$ with the regularized case $\lambda_{\mathrm{space}}=3$. The color scale reports the pointwise error $\widehat{\sigma}_{\mathrm{tree}}-\sigma_{\mathrm{loc}}$ at the tree nodes, plotted in the maturity/spot-ratio plane. The same color scale is used in both panels. $N_T=720$ and $N_E=8000$. }
	\end{figure}
	\begin{figure}[htbp]
		\centering
		
		\begin{subfigure}[t]{0.95\textwidth}
			\centering
			\includegraphics[width=\textwidth]{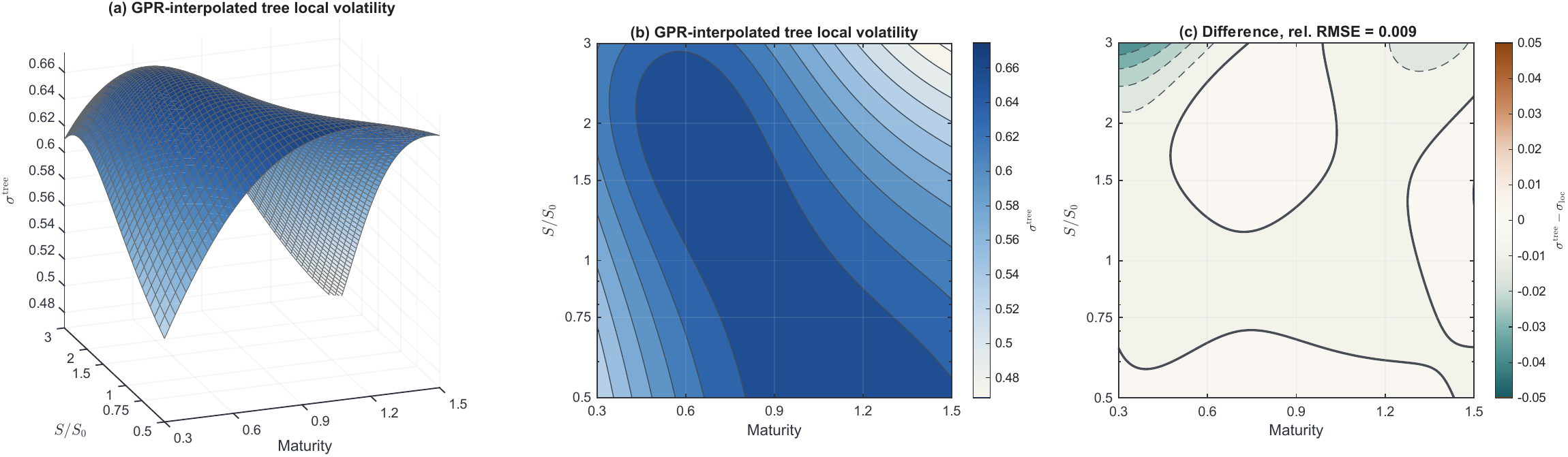}
			\caption{$10\times 20$ training grid.}
		\end{subfigure}
		
		\vspace{0.75em}
		
		\begin{subfigure}[t]{0.95\textwidth}
			\centering
			\includegraphics[width=\textwidth]{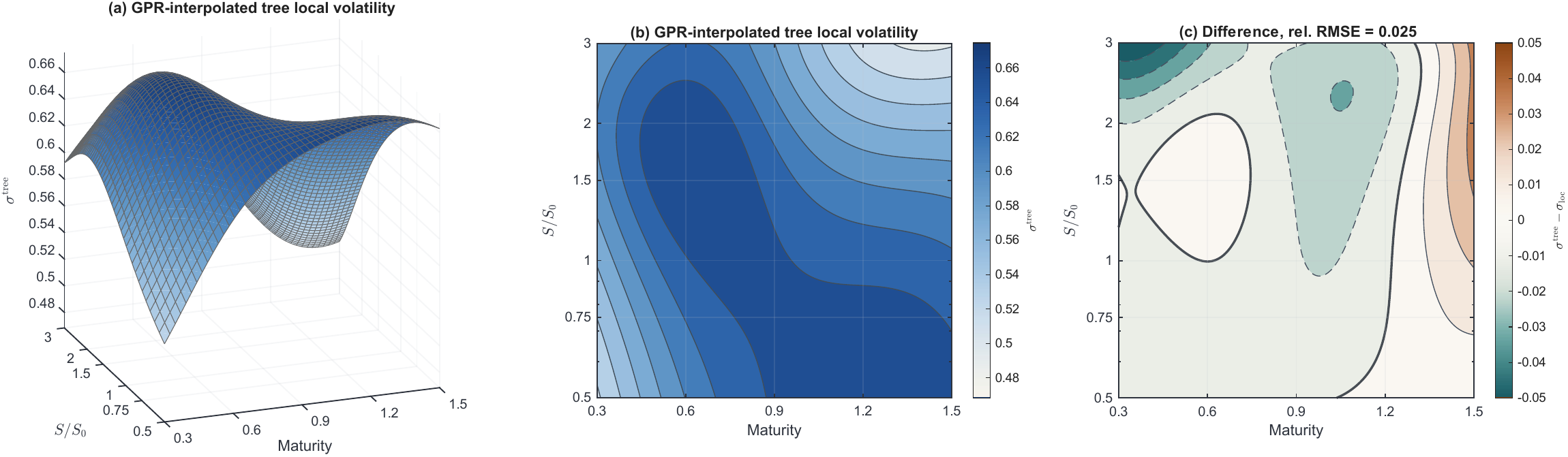}
			\caption{$3\times 6$ training grid.}
		\end{subfigure}
		
	\caption{\label{fig:lambda_lv_surfaces}Comparison between the true local-volatility surface and the
		GPR-reconstructed tree-implied local-volatility surface for representative
		regularized configurations with \(N_T=720\), \(N_E=8000\), and
		\(\lambda_{\mathrm{space}}=3\). For each configuration, panel (a) shows the
		GPR-reconstructed tree-implied local-volatility surface, panel (b) displays the
		same surface as a contour plot, and panel (c) reports the difference
		\(\treevol-\locvol\) over the common validation domain. In panel (c), dashed
		contour lines correspond to negative levels, solid contour lines to positive
		levels, and the thicker solid contour denotes the zero level.}
		
	\end{figure}
	\FloatBarrier
    
\subsection{SPX put options}

We next turn to a market-data experiment on SPX European put options. To enable a
direct comparison with the DSCL benchmark of \citet{wang2025}, we use the same
SPX dataset considered there, dated May 18, 2019.
As in that paper, dividends have been removed and the dataset is split into training and testing subsets containing
1720 and 1725 option prices, respectively, and we used \(N_E=16\,000\) cumulative training epochs for the calibration. In all cases, we use linear interpolation in maturity whenever the observed maturities do not coincide with tree dates. The main features of this market dataset are summarized in \Cref{tab:spx-data}.

	\begin{table}
		\centering
		\begin{tabular*}{\textwidth}{@{\extracolsep{\fill}}lllllll@{}}
			\toprule
			Symbol & Meaning & Value & \quad & Symbol & Meaning & Value\tabularnewline
			\midrule
			$S_0$ & spot price & $2859.53$ & & $N_{\mathrm{train}}$ & training option prices & $1720$\tabularnewline
			$r$ & risk-free rate & $0.023$ & & $N_{\mathrm{test}}$ & testing option prices & $1725$\tabularnewline
			$T$ & maturity range & $[0.055,\,2.5]$ & & $K$ & strike range & $[1150,\,4000]$\tabularnewline
			%Date & market snapshot & May 18, 2019 & & Option type & quoted contracts & SPX European puts\tabularnewline
			\bottomrule
		\end{tabular*}
		\caption{\label{tab:spx-data}Main characteristics of the SPX market dataset used in the empirical comparison with \citet{wang2025}.}
        \label{tab:dataSPX}
	\end{table}
	
	\begin{table}[t]
		\begin{centering}
					\small{\begin{tabular*}{\textwidth}{@{\extracolsep{\fill}}llccccccccccccccccccccc@{}}
				\toprule
				$\hfill N_{T}$ &
				& \multicolumn{4}{c}{$100$}
				& & \multicolumn{4}{c}{$200$}
				& & \multicolumn{4}{c}{$400$}
				& & \multicolumn{4}{c}{$800$}
				& & \tabularnewline
				
				$\hfill N_{E}\ (\times10^3)$ &
				&   $2$ & $4$ & $8$ & $16$
				& &$2$ & $4$ & $8$ & $16$
				&  &$2$ & $4$ & $8$ & $16$
				&  & $2$ & $4$ & $8$ & $16$
				& & DSCL \tabularnewline
				
				\cmidrule{3-6}\cmidrule{8-11}\cmidrule{13-16}\cmidrule{18-21}\cmidrule{23-23}
				
				OP IS  ($\times10^{-2}$) &
				& \nz{97.19}& \nz{91.91}& \nz{87.75}& \nz{84.22}
				& & \nz{62.42}& \nz{54.59}& \nz{49.21}& \nz{44.76}
				& & \nz{52.99}& \nz{44.48}& \nz{37.36}& \nz{32.18}
				& & \nz{52.42}& \nz{41.94}& \nz{34.23}& \nz{29.97}
				& & $97$ \tabularnewline
				
				OP OOS  ($\times10^{-2}$) &
				& \nz{104.30}& \nz{98.29}& \nz{94.28}& \nz{92.08}
				& & \nz{69.55}& \nz{60.59}& \nz{55.69}& \nz{52.21}
				& & \nz{61.67}& \nz{48.86}& \nz{40.86}& \nz{35.97}
				& & \nz{59.33}& \nz{45.87}& \nz{37.39}& \nz{33.95}
				& & $326$ \tabularnewline
 
							RT  ($\times10^3$) 
				& & \nl{57}& \nl{90}& \nl{143}& \nl{249}
				& & \nl{164}& \nl{290}& \nl{522}& \nl{988}
				& & \nl{478}& \nl{831}& \nl{1526}& \nl{2894}
				& & \nl{1416}& \nl{2469}& \nl{4555}& \nl{8678}
				&   \tabularnewline
				
				\bottomrule
		\end{tabular*}}
			\par\end{centering}
		\caption{\label{tab:spx_rmse}
			Option-pricing RMSEs,  expressed in units of $10^{-2}$, for the SPX experiment, reported in-sample (OP IS) and out-of-sample (OP OOS) for the proposed method at different tree sizes $N_T$ and cumulative training epochs $N_E$, with $\lambda_{\mathrm{space}}=0$. The last column reports the benchmark values from Wang et al. ~\cite{wang2025}.
		}
	\end{table}
	
	Table~\ref{tab:spx_rmse} reports the SPX pricing errors obtained with the
	proposed method for cumulative training budgets
	\(N_E\in\{2000,4000,8000,16000\}\) and with spatial regularization switched
	off. The results show a clear improvement as the number of tree time steps
	increases: both the in-sample and out-of-sample RMSE decrease substantially
	when passing from \(N_T=100\) to \(N_T=800\). The comparison with Wang et al.\
	is favorable. In particular, the proposed method yields lower out-of-sample RMSEs than the DSCL benchmark for all reported values of \(N_T\), while the in-sample metric improves markedly as the tree is refined.  Increasing $N_T$ produces a substantial reduction in both in-sample and out-of-sample RMSE, while the gains from increasing $N_E$ beyond the first stages are relatively moderate. This shows that  tree resolution is more critical than prolonged training once the main deformation of the base tree has already been learned.

	To complement the RMSE values reported in Table~\ref{tab:spx_rmse},  Figure~\ref{fig:spx_pricing_errors} displays pricing errors and percent errors for the calibrated tree with $N_T=800$ and $N_E=16\,000$, separately for in-sample and out-of-sample quotes over the maturity--moneyness plane. The top row reports pricing errors and the bottom row the corresponding relative errors as a percentage. In each row, the left panel refers to in-sample quotes and the right panel to out-of-sample quotes. 
	
	The in-sample and out-of-sample pricing-error maps broadly show a similar structure, which suggests that the calibrated tree does not merely overfit the training quotes but captures a stable pricing pattern across the observed SPX surface. The larger pricing errors tend to concentrate in the short-maturity region and in the more extreme moneyness zones. 
	\begin{figure} 
				\centering
		\begin{subfigure}[t]{1\textwidth}
			\centering
			\includegraphics[width=0.98\textwidth]{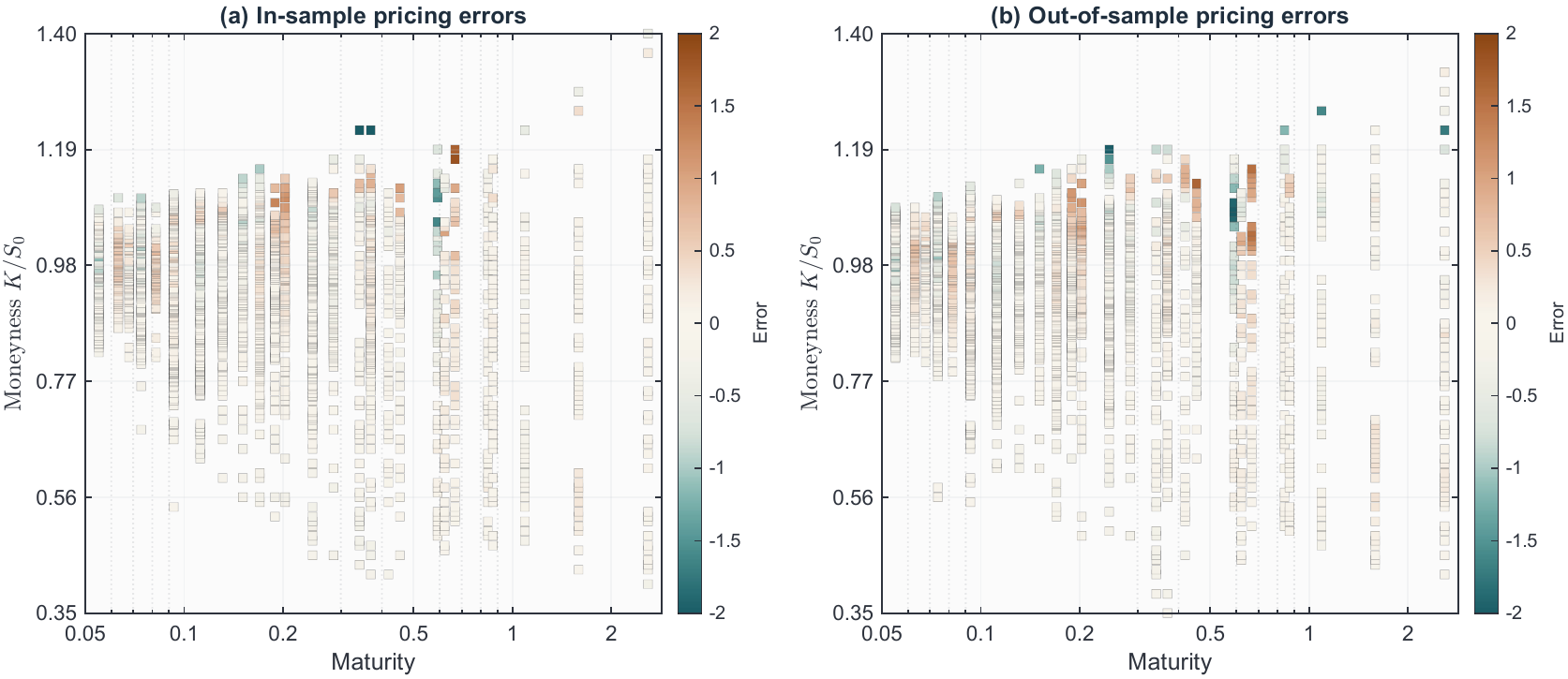} 
			\caption{Error}
			\label{fig:heatmap1020}
		\end{subfigure}\hfill
		\begin{subfigure}[t]{1\textwidth}
			\centering 
			\includegraphics[width=0.98\textwidth]{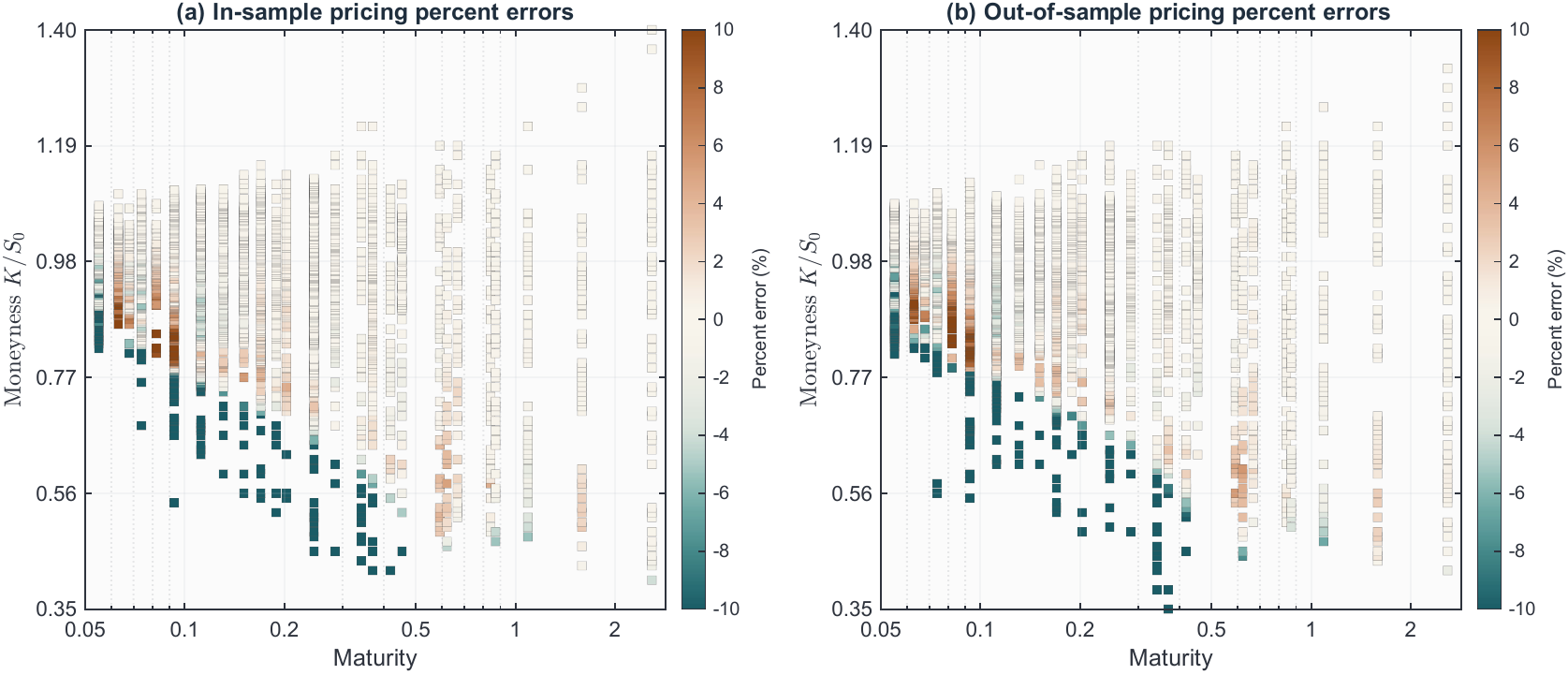}
			\caption{Percent error}
			\label{fig:heatmap0306}
		\end{subfigure}		  
		 \caption{Pricing errors and percent errors for the SPX experiment with $N_T=800$ and $N_E=16\,000$. The top row reports pricing errors for in-sample quotes (left) and out-of-sample quotes (right), while the bottom row reports the corresponding percent errors. Each square corresponds to one market option quote, positioned by maturity and moneyness. The color scales are fixed across the error panels and across the percent-error panels.}\label{fig:spx_pricing_errors}
	\end{figure}

	\Cref{fig:spx_lv_0e0} displays the local-volatility structure implied by the calibrated SPX tree for the representative case $N_T=800$ and $N_E=16\,000$.
    % Since no reference local-volatility surface is available in market data, the figure should be interpreted as a qualitative diagnostic rather than as a pointwise validation tool.
    The left panel shows the local-volatility estimates extracted directly from the calibrated tree, the central panel shows  the estimated local volatility surface, and the right panel displays the corresponding GPR-reconstructed contour plot.  The reconstructed surface is quite regular and its  shape is close to the SPX local-volatility surface displayed by \citet{wang2025}, so the proposed method produces a market-implied volatility that is broadly consistent with that benchmark.
	 
	\begin{figure} 
		\centering
		\includegraphics[width=0.95\textwidth]{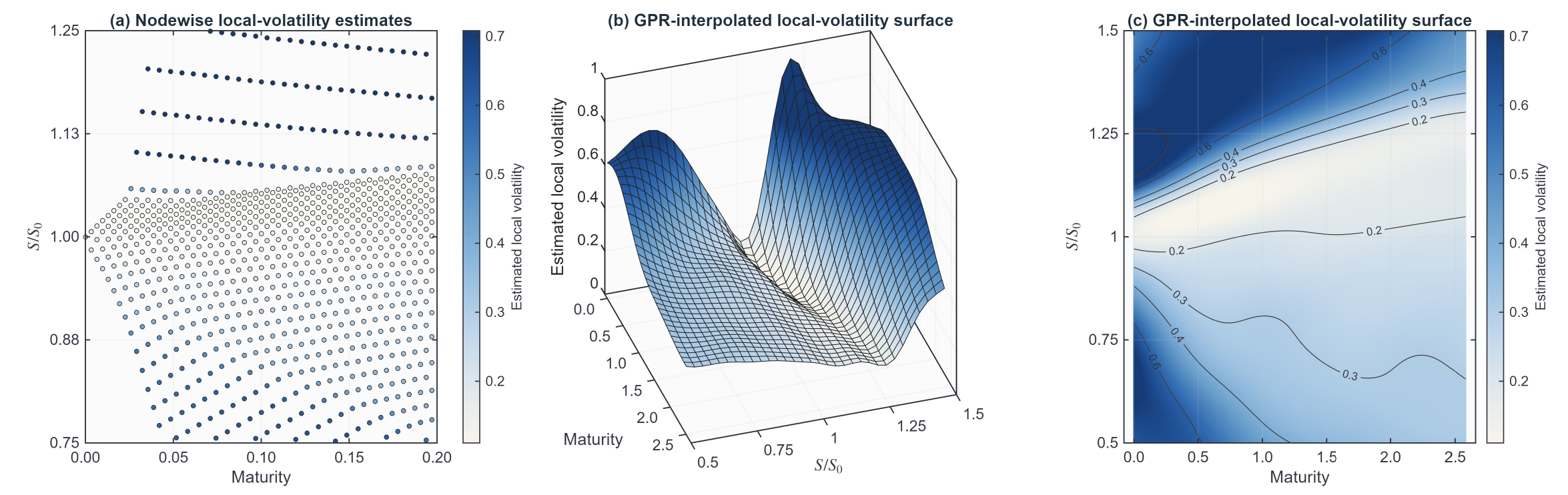} 
		\caption{\label{fig:spx_lv_0e0}Local-volatility diagnostic for the SPX experiment with $N_T=800$, $N_E=16\,000$ and $\lambda_{\mathrm{space}}=0$. Left: nodes of the calibrated tree, with color indicating the corresponding tree-implied local-volatility estimate at each node. Center and right: 3D and contour representations of the GPR-reconstructed local-volatility surface. }
	\end{figure}

As a final robustness check, Table~\ref{tab:spx_rmse_lambda} reports the SPX
pricing errors obtained by varying the spatial regularization coefficient while
keeping \(N_E=16\,000\) fixed. The behavior is very regular: both in-sample and
out-of-sample RMSE values decrease as the number of time steps increases, and the
dependence on \(\lambda_{\mathrm{space}}\) is mild. Compared with the
unregularized results, spatial regularization produces only marginal changes in
pricing accuracy. The best regularized result is very close to the
unregularized one, which suggests that, for the large and informative SPX
dataset, the calibration is already stable without spatial penalization. In
such settings, the natural default is therefore \(\lambda_{\mathrm{space}}=0\), but the results improve or do not deteriorate too much when it is taken differently from zero.

		\begin{table}[t]
		\begin{centering}
			\small{\begin{tabular*}{\textwidth}{@{\extracolsep{\fill}}llccccccccccccccccccccc@{}}
					\toprule
					$\hfill N_{T}=$ &
					& \multicolumn{4}{c}{$100$}
					& & \multicolumn{4}{c}{$200$}
					& & \multicolumn{4}{c}{$400$}
					& & \multicolumn{4}{c}{$800$}
					& & \tabularnewline
					
					$\hfill \lambda_{\mathrm{space}}=$ &
					& $0$ & $0.01$ & $0.03$ & $0.1$
					& & $0$ & $0.01$ & $0.03$ & $0.1$
					& & $0$ & $0.01$ & $0.03$ & $0.1$
					& & $0$ & $0.01$ & $0.03$ & $0.1$
					& & DSCL \tabularnewline
					
					\cmidrule{3-6}\cmidrule{8-11}\cmidrule{13-16}\cmidrule{18-21}\cmidrule{23-23} 
					
					OP IS ($\times10^{-2}$)&
					& \nz{84.22}&\gc \nz{85.26}& \nz{85.31}& \nz{83.95}
					& & \nz{44.76}& \nz{44.21}& \gc\nz{44.07}& \nz{45.37}
					& & \nz{32.18}& \gc\nz{32.25}& \nz{32.62}& \nz{34.02}
					& & \nz{29.97}& \nz{30.37}& \gc\nz{29.63}& \nz{31.66}
					& & $97$ \tabularnewline
					
					OP OOS ($\times10^{-2}$)&
					& \nz{92.08}& \gc\nz{91.29}& \nz{91.40}& \nz{92.71}
					& & \nz{52.21}& \nz{51.08}& \gc\nz{50.35}& \nz{52.74}
					& & \nz{35.97}& \gc\nz{35.84}& \nz{36.37}& \nz{39.47}
					& & \nz{33.95}& \nz{34.58}& \gc\nz{33.94}& \nz{35.50}
					& & $326$ \tabularnewline
						RT  ($\times10^3$) 
					& & \nl{249}& \gc\nl{349}& \nl{387}& \nl{396}
					& & \nl{988}& \nl{759}&  \gc\nl{820}& \nl{862}
					& & \nl{2894}&  \gc\nl{2170}& \nl{2251}& \nl{2383}
					& & \nl{8678}& \nl{7196}&\gc \nl{7418}& \nl{7527}
					&   \tabularnewline
					
					\bottomrule
			\end{tabular*}}
			\par\end{centering}
		\caption{\label{tab:spx_rmse_lambda}
			Option-pricing RMSEs, expressed in units of $10^{-2}$,  for the SPX experiment, reported in-sample (OP IS) and out-of-sample (OP OOS) for the proposed method at different tree sizes $N_T$ and for different values of the spatial regularization coefficient $\lambda_{\mathrm{space}}$, with cumulative training budget fixed at $N_E=16\,000$. The last column reports the DSCL benchmark from Wang et al.~\cite{wang2025}. The corresponding unregularized results are reported in Table~\ref{tab:spx_rmse}.}
	\end{table}

		\begin{figure} 
		\centering
		\includegraphics[width=0.95\textwidth]{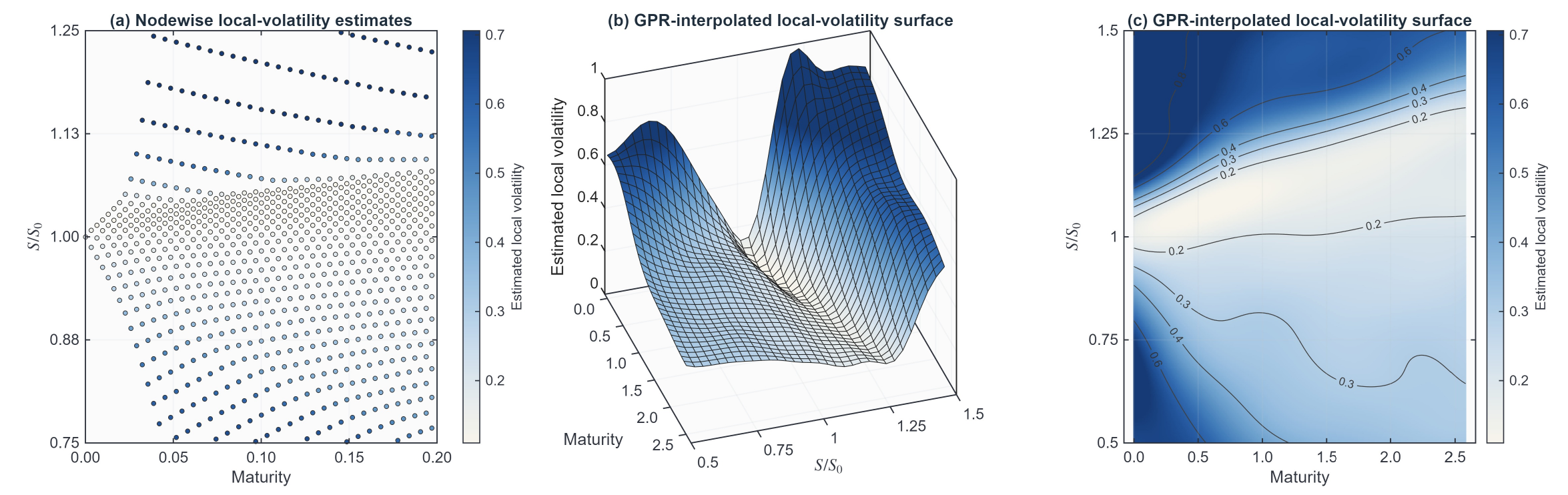} 
		\caption{\label{fig:spx_lv_3em2}Local-volatility diagnostic for the SPX experiment with $N_T=800$, $N_E=16\,000$ and $\lambda_{\mathrm{space}}=0.03$. Left: nodes of the calibrated tree, with color indicating the corresponding tree-implied local-volatility estimate at each node. Center and right: 3D and contour representations of the GPR-reconstructed local-volatility surface. }
	\end{figure}
	
	%\FloatBarrier
	\subsection{American put options and early exercise}
	\label{subsec:american_options}
	
	\subsubsection{Post-calibration valuation of American puts}
	\label{subsubsec:american_post_calibration}
	
	One practical advantage of the proposed methodology is that the resulting recombining tree can be reused to value new contracts by standard backward induction once calibration has been completed,  without solving a PDE or training a new neural model. This feature is particularly relevant for derivatives with early-exercise features, such as American options.
	
	To illustrate this point, we consider the trees calibrated in the synthetic-data experiment using European calls, and reuse them to price American put options. We compare the prices to the ones generated by solving a PDE with free boundary. We use a dense out-of-sample grid with $256$ maturities and $256$ strikes, with maturities uniformly distributed in $[0.3,1.5]$ and strikes uniformly distributed in $[500,3000]$.

	The RMSE values reported in
	Table~\ref{tab:am_perfmap} show
   that the
	 American-put valuation errors are not uniformly distributed over the
	 maturity-moneyness domain, but are concentrated in localized regions, especially
	 for the sparse $3\times6$ training grid. The dominant source of improvement is the refinement
	of the tree in time: increasing $N_T$ produces a substantial reduction in
	out-of-sample pricing RMSE for both training configurations, whereas the gains
	from increasing $N_E$ become comparatively milder once the main deformation of
	the benchmark tree has already been learned. In particular, the $10\times20$
	training grid exhibits a marked improvement as $N_T$ increases, while the
	$3\times6$ grid remains significantly less accurate.

	\begin{table} 
		\centering
			\small{\begin{tabular*}{\textwidth}{@{\extracolsep{\fill}}llccccccccccccccc@{}}
				\toprule
				$\hfill N_{T}=$ & 
				& \multicolumn{4}{c}{$180$}
				& & \multicolumn{4}{c}{$360$}
				& & \multicolumn{4}{c}{$720$} 
				& \tabularnewline
				
				$\hfill N_{E}=$ &
				& $1000$ & $2000$ & $4000$ & $8000$
				& & $1000$ & $2000$ & $4000$ & $8000$
				& & $1000$ & $2000$ & $4000$ & $8000$ 
				&  \tabularnewline
				
				\cmidrule{3-6}\cmidrule{8-11}\cmidrule{13-16}
				
				\multicolumn{12}{l}{\rule{0pt}{5mm}Input grid $N_{T}^{opt}\times N_{K}^{opt}=10\times20$ European call options}
				& & & &  \tabularnewline

				OP OOS  ($\times10^{-2}$) &
				& \nz{23.82}& \nz{22.68}& \nz{22.39}& \nz{22.39}
				& & \nz{14.61}& \nz{12.47}& \nz{12.35}& \nz{12.42}
				& & \nz{11.49}& \nz{10.33}& \nz{9.90}& \nz{10.36}
				&  \tabularnewline

				\multicolumn{12}{l}{\rule{0pt}{5mm}Input grid $N_{T}^{opt}\times N_{K}^{opt}=3\times6$  European call options}
				& & & &  \tabularnewline

				OP OOS  ($\times10^{-2}$) &
				
				& \nz{162.11}& \nz{163.67}& \nz{163.99}& \nz{167.58}
				& & \nz{98.57}& \nz{95.54}& \nz{100.12}& \nz{106.76}
				& & \nz{104.12}& \nz{100.54}& \nz{100.46}& \nz{93.85}
				& \tabularnewline 
				
				\bottomrule
		\end{tabular*}}
		\caption{Out-of-sample option-pricing RMSE, expressed in units of $10^{-2}$,  for American put options as a function of the number of tree time steps $N_T$ and the cumulative number of training epochs $N_E$, for the two synthetic training grids made of European call options. $\lambda_{\mathrm{space}}=0$. }
		\label{tab:am_perfmap}
	\end{table}
	
	The effect of spatial regularization on this post-calibration American-put
	valuation test is reported in Table~\ref{tab:lambda_pricing_am_eu}. The dense
	$10\times20$ training grid is already sufficiently informative, and the
	regularization has only a limited effect on the American put out-of-sample
	errors. In this case, the best values are obtained either without regularization
	or with a moderate value of $\lambda_{\mathrm{space}}$, depending on the tree
	resolution. By contrast, the sparse $3\times6$ grid benefits substantially from
	a positive spatial penalty. For all reported values of $N_T$, suitable
	positive values of $\lambda_{\mathrm{space}}$ lead to a sizeable reduction in
	the out-of-sample RMSE. This again confirms that the spatial regularization is
	   useful when the calibration data are too sparse to
	identify a stable tree deformation.
	
		\begin{table}
		\begin{centering}
			\small{
				\begin{tabular*}{\textwidth}{@{\extracolsep{\fill}}llcccccccccccccccccc@{}}
					\toprule
					$\hfill N_{T}=$
					& & \multicolumn{5}{c}{$180$}
					& & \multicolumn{5}{c}{$360$}
					& & \multicolumn{5}{c}{$720$}
					&    \tabularnewline
					
					$\hfill \lambda_{\mathrm{space}}=$
					& & $0$ & $3$ & $10$ & $30$ & $100$
					& & $0$ & $3$ & $10$ & $30$ & $100$
					& & $0$ & $3$ & $10$ & $30$ & $100$
					&    \tabularnewline
					
					\cmidrule{3-7}\cmidrule{9-13}\cmidrule{15-19}
					
					\multicolumn{19}{l}{\rule{0pt}{5mm}Input grid $N_{T}^{opt}\times N_{K}^{opt}=10\times20$   European call options}
					&    \tabularnewline

					OP OOS ($\times10^{-2}$) 
					&& \nz{22.39}& \nz{22.92}& \gc\nz{19.97}& \nz{22.57}& \nz{32.11}
					&& \gc\nz{12.42}& \nz{16.60}& \nz{20.67}& \nz{21.40}& \nz{30.09}
					&& \gc\nz{10.36}& \nz{15.42}& \nz{21.90}& \nz{26.47}& \nz{29.08}
					&    \tabularnewline 
					\multicolumn{19}{l}{\rule{0pt}{5mm}Input grid $N_{T}^{opt}\times N_{K}^{opt}=3\times6$   European call options}
					&    \tabularnewline

					OP OOS ($\times10^{-2}$) 
					&& \nz{167.62}& \nz{97.58}& \nz{86.73}& \nz{86.52}&\gc \nz{82.12}
					&& \nz{106.81}& \nz{76.32}&\gc \nz{62.95}& \nz{64.25}& \nz{73.91}
					&& \nz{93.92}&\gc \nz{56.09}& \nz{67.54}& \nz{62.78}& \nz{67.80}
					&  \tabularnewline
					 \bottomrule
				\end{tabular*}
			}
			\par\end{centering} 
		\caption{Out-of-sample option-pricing RMSE for American put options when the
			tree is trained on synthetic European call option prices. Results are reported
			for two training grids, different tree resolutions $N_T$, and different values
			of  $\lambda_{\mathrm{space}}$. Entries are
			expressed in units of $10^{-2}$; highlighted values denote the best result
			within each $N_T$ block. All calibrations use $N_E=8000$ training epochs.}
		\label{tab:lambda_pricing_am_eu} 
	\end{table}

\subsubsection{Direct calibration on American put prices}
\label{subsubsec:american_direct_calibration}

We now consider a second experiment involving American put options. Differently
from Section~\ref{subsubsec:american_post_calibration}, where the tree is
calibrated on European call prices and subsequently reused to price American
puts, now the objective function for calibration  is built from American put prices with the same maturities and strikes as in Section~\ref{sec:sy_da_ex}.

		\begin{table}
		\begin{centering}
			\small{\begin{tabular*}{\textwidth}{@{\extracolsep{\fill}}llccccccccccccccc@{}}
					\toprule
					$\hfill N_{T}=$ 
					& & \multicolumn{4}{c}{$180$}
					& & \multicolumn{4}{c}{$360$}
					& & \multicolumn{4}{c}{$720$}
					&   \tabularnewline
					
					$\hfill N_{E}=$ 
					& & $ {1000}$ & $2000$ & $4000$ & $8000$
					& & $1000$ & $ {2000}$ & $4000$ & $8000$
					& & $1000$ & $ {2000}$ & $ {4000}$ & $8000$
					& \tabularnewline
					
					\cmidrule{3-6}\cmidrule{8-11}\cmidrule{13-16} 
					
					\multicolumn{17}{l}{\rule{0pt}{5mm}Input grid $N_{T}^{opt}\times N_{K}^{opt}=10\times20$  American put options}
					\tabularnewline
					
					\midrule
					
					OP IS ($\times10^{-2}$)  
					& & \nz{16.16}& \nz{10.81}& \nz{9.30}& \nz{7.95}
					& & \nz{12.07}& \nz{7.14}& \nz{5.10}& \nz{3.91}
					& & \nz{11.72}& \nz{5.54}& \nz{4.05}& \nz{3.16}
					&   \tabularnewline
					
					OP OOS ($\times10^{-2}$) 
					& & \nz{23.63}& \nz{21.56}& \nz{21.66}& \nz{21.87}
					& & \nz{12.57}& \nz{10.16}& \nz{10.12}& \nz{10.41}
					& & \nz{10.82}& \nz{6.38}& \nz{5.46}& \nz{5.08}
					&    \tabularnewline

					RT ($\times10^3$) 
					& & \nl{335}& \nl{654}& \nl{1270}& \nl{2502}
					& & \nl{700}& \nl{1397}& \nl{2791}& \nl{5570}
					& & \nl{1912}& \nl{3898}& \nl{7918}& \nl{15985}
					&   \tabularnewline
					
					\multicolumn{17}{l}{\rule{0pt}{5mm}Input grid $N_{T}^{opt}\times N_{K}^{opt}=3\times6$  American put options}
					 \tabularnewline
					
					\midrule
					
					OP IS ($\times10^{-2}$) 
					& & \nz{5.91}& \nz{2.51}& \nz{1.35}& \no{0.35}
					& & \nz{9.51}& \nz{3.49}& \nz{1.11}& \no{0.23}
					& & \nz{5.43}& \nz{2.11}& \nz{1.51}& \no{0.21}
					&    \tabularnewline
					
					OP OOS ($\times10^{-2}$) 
					& & \nz{134.01}& \nz{133.69}& \nz{134.41}& \nz{133.25}
					& & \nz{109.89}& \nz{112.14}& \nz{113.54}& \nz{114.19}
					& & \nz{113.81}& \nz{110.29}& \nz{109.91}& \nz{119.41}
					&  \tabularnewline

					RT ($\times10^3$) 
					& & \nl{312}& \nl{616}& \nl{1222}& \nl{2432}
					& & \nl{684}& \nl{1345}& \nl{2663}& \nl{5298}
					& & \nl{1733}& \nl{3424}& \nl{6790}& \nl{13529}
					&    \tabularnewline
										
					\bottomrule
			\end{tabular*}}

			\par\end{centering}
		\caption{\label{tab:RMSE_am_put}RMSE for American option pricing (OP)  in-sample (IS) and out-of-sample (OOS), expressed in units of $10^{-2}$.  Runtime (RT) is
			reported in units of \(10^3\) seconds. Both training and testing are performed on American put options. $\lambda_{\mathrm{space}}=0$.}
	\end{table}

The results in Table~\ref{tab:RMSE_am_put} confirm that the proposed calibration
procedure can also be trained directly on American put prices. The qualitative behavior is similar to that observed in the European-call
experiments. Increasing the number of tree time steps improves pricing accuracy
clearly on the dense grid, while the sparse grid remains more sensitive and less
monotone out of sample. Increasing the training budget mainly refines the
in-sample fit once the dominant deformation of the base tree has already been
learned.

For the dense \(10\times20\) grid, direct calibration on American put prices
gives accurate results. The out-of-sample RMSE decreases from
\(22\times10^{-2}\) for \(N_T=180\) and \(N_E=8000\) to
\(5\times10^{-2}\) for \(N_T=720\) and \(N_E=8000\). This is substantially
better than the American-put valuation test reported in
Table~\ref{tab:am_perfmap}, where the tree was calibrated on European calls
and then reused to price American puts. In that case, the corresponding
out-of-sample RMSE for the \(10\times20\) grid and \(N_T=720\), \(N_E=8000\)
was about \(10\times10^{-2}\). Thus, when the objective is the accurate
valuation of American puts, including the early-exercise feature directly in
the calibration roughly halves the out-of-sample pricing error in the
dense-grid case.

The same comparison is more nuanced for the sparse \(3\times6\) grid. Direct
training on American puts produces very small in-sample errors, reaching
\(0.2\times10^{-2}\) for \(N_T=720\) and \(N_E=8000\), but the out-of-sample
RMSE remains much larger, around \(119\times10^{-2}\). This is close to,
and slightly worse than, the corresponding value in Table~\ref{tab:am_perfmap},
where the tree calibrated on European calls achieved an out-of-sample RMSE of
about \(94\times10^{-2}\) for the same tree size and training budget. This
confirms that, under sparse training information, fitting American prices
directly can lead to a very accurate reproduction of the observed quotes
without necessarily improving global out-of-sample performance.

The main drawback of direct calibration on American put prices is its higher computational cost. Compared with the
European-call calibration in Table~\ref{tab:RMSE_1}, runtimes are substantially
larger, because each loss evaluation requires a backward induction with the
early-exercise obstacle rather than a simple Arrow--Debreu aggregation at
maturity. For example, in the dense \(10\times20\) case with \(N_T=720\) and
\(N_E=8000\), the runtime increases from about \(3.6\times10^3\) seconds in
the European-call experiment to about \(16.0\times10^3\) seconds when training
is performed directly on American puts. A similar pattern is observed for the
sparse grid, where the runtime rises to about \(13.5\times10^3\) seconds.

	\begin{table}
	\begin{centering}
		\small{
			\begin{tabular*}{\textwidth}{@{\extracolsep{\fill}}llcccccccccccccccccc@{}}
				\toprule
				$\hfill N_{T}=$
				& & \multicolumn{5}{c}{$180$}
				& & \multicolumn{5}{c}{$360$}
				& & \multicolumn{5}{c}{$720$}
				&    \tabularnewline
				
				$\hfill \lambda_{\mathrm{space}}=$
				& & $0$ & $3$ & $10$ & $30$ & $100$
				& & $0$ & $3$ & $10$ & $30$ & $100$
				& & $0$ & $3$ & $10$ & $30$ & $100$
				&    \tabularnewline
				
				\cmidrule{3-7}\cmidrule{9-13}\cmidrule{15-19}
				
				\multicolumn{19}{l}{\rule{0pt}{5mm}Input grid $N_{T}^{opt}\times N_{K}^{opt}=10\times20$   American put options}
				&    \tabularnewline
				
				\midrule
				
				OP IS ($\times10^{-2}$) 
				&& \nz{7.95}& \nz{10.14}& \nz{11.38}&  \gc\nz{13.99}& \nz{18.46}
				&& \nz{3.91}&  \gc\nz{4.55}& \nz{6.81}& \nz{8.58}& \nz{13.08}
				&&  \nz{3.16}& \gc \nz{4.35}& \nz{6.01}& \nz{8.58}& \nz{12.66} 
				&    \tabularnewline
				
				OP OOS ($\times10^{-2}$) 
				&& \nz{21.87}& \nz{20.21}& \nz{19.16}& \gc\nz{18.80}& \nz{19.58}
				&& \nz{10.41}& \gc \nz{8.99}& \nz{9.07}& \nz{10.15}& \nz{13.79}
				&& \nz{5.08}& \gc \nz{4.87}& \nz{6.73}& \nz{7.82}& \nz{12.15}
				&    \tabularnewline 
				RT ($\times10^3$)   
				&& \nl{2502}& \nl{3088}& \nl{3131}& \gc \nl{3180}& \nl{3207}
				&& \nl{5570}&\gc \nl{7275}& \nl{7465}& \nl{7548}& \nl{7624}
				&& \nl{15985}& \gc \nl{19497}& \nl{19228}& \nl{19301}& \nl{19256}
				&    \tabularnewline
				
				\multicolumn{19}{l}{\rule{0pt}{5mm}Input grid $N_{T}^{opt}\times N_{K}^{opt}=3\times6$   American put options}
				&    \tabularnewline
				
				\midrule
				
				OP IS ($\times10^{-2}$) 
				& &\no{0.35} & \nz{1.81}&\gc \nz{3.08}& \nz{5.35}& \nz{8.71}
				& &\no{0.23} &\nz{2.26}& \nz{3.83}& \nz{5.45}&  \gc\nz{10.22}
				& &\no{0.21} & \nz{2.32}& \nz{4.31}&\gc \nz{6.60}& \nz{11.75}
				&    \tabularnewline
				
				OP OOS ($\times10^{-2}$) 
				& &\nz{133.25} & \nz{60.99}&\gc \nz{51.26}& \nz{56.51}& \nz{65.76}
				& &\nz{114.19} & \nz{62.92}& \nz{59.12}& \nz{54.94}&\gc \nz{50.19}
				& &\nz{119.41} & \nz{59.66}& \nz{57.95}&\gc \nz{47.89}& \nz{55.30}
				&  \tabularnewline
				
				RT ($\times10^3$)   
				&& \nl{2432}& \nl{2983}& \gc\nl{3028}& \nl{3000}& \nl{3026}
				&& \nl{5298}& \nl{6623}& \nl{6829}& \nl{6780}&\gc \nl{6882}
				&& \nl{13529}& \nl{17255}& \nl{17174}&\gc \nl{17304}&\nl{17214}
				& \tabularnewline\bottomrule 
			\end{tabular*}
		}
		\par\end{centering} 
	\caption{\label{tab:lambda_pricing_am_am} 
		RMSE for American-put option pricing (OP), reported in-sample (IS)
		and out-of-sample (OOS), together with runtime (RT), for the method trained
		directly on synthetic American put option prices. Results are reported for two
		American-put training grids, different tree resolutions $N_T$, and different
		values of the spatial-regularization coefficient $\lambda_{\mathrm{space}}$,
		with cumulative training budget fixed at $N_E=8000$. Pricing errors are
		expressed in units of $10^{-2}$, while runtime is expressed in units of
		$10^3$ seconds. Grey-shaded entries identify, for each input grid and each
		fixed value of $N_T$, the value of $\lambda_{\mathrm{space}}$ that minimizes
		the OP OOS error; the corresponding OP IS and RT values are shaded as well.}
	
\end{table}

	Table~\ref{tab:lambda_pricing_am_am} reports the effect of spatial
	regularization when the tree is trained directly on American put option prices,
	for \(N_E=8000\). The results indicate that the impact of
	\(\lambda_{\mathrm{space}}\) is again modest for the dense \(10\times20\) grid, while
	it is more pronounced for the sparse \(3\times6\) grid, where positive
	regularization substantially improves out-of-sample performance.

	Overall, these results indicate that direct calibration on American option
	prices is feasible and can significantly improve pricing accuracy when the
	training grid is sufficiently informative. At the same time, the experiment
	also highlights the computational cost of incorporating early exercise instruments.
	
	\FloatBarrier
	\section{Conclusion}
	\label{sec:conclusion}
	
	We have proposed a neural calibration framework for constructing a recombining
	binomial tree directly from option prices. By deforming a benchmark tree, it produces a complete and arbitrage-free model in discrete time that closely matches
	market prices. That model can then be used for the pricing and hedging of other contingent claims, including options with early exercise features, without referring to the machine learning algorithms that were used to design it.
	An optional spatial regularizer can be added when the option data are sparse and
	the inverse problem is weakly identified, and we show that this hardly changes the results for cases where such regularization is not needed.
    	The method does not aim for a calibration of a continuous local-volatility function, and experiments show that when the available option dataset is sparse, the recovery of a regular local-volatility surface remains challenging.
	
	The proposed methodology occupies an intermediate position between classical calibrated lattice models and recent deep-learning approaches to option surfaces. Its main advantage is that it returns a discrete pricing model rather than an intermediate machine learning object, which is attractive when the final objective is fast and reliable valuation of both prices and hedge strategies. Moreover, the method requires only a small number of initial design choices, mainly
the tree resolution \(N_T\) and the maximum training budget \(N_E\).

       Our approach can be interpreted as a design problem for stock price paths under a recombination constraint, and as such it combines the flexibility
	of neural calibration with the easy interpretability and computational convenience of
	lattice models for pricing and hedging.

\section*{Acknowledgements}

The authors report that no generative AI has been used for this paper, apart from suggestions for improvements in language and grammar (ChatGPT  5.6 Sol). The authors take full responsibility for the content of the publication.

	\FloatBarrier
	
	%========================================================
	% Appendix
	%========================================================

    \newpage 
%========================================================
% Appendix
%========================================================
\appendix

\titleformat{\section}
  {\normalfont\Large\bfseries}
  {Appendix \thesection:}{0.5em}{}

\section{Proofs}
\label{ap:proofs}

\subsection{Proof of Proposition~\ref{prop:arbfree}}
\label{app:proof-arbfree}

\begin{proof}
	The assumed inequalities imply that the numerator and denominator in \Cref{eq:rnp} are strictly positive and that the numerator is strictly smaller than the denominator, so \(p_{n,j}\in(0,1)\). Local, and thus global, absence of arbitrage follows from the existence of a unique one-step equivalent martingale measure. Completeness follows from the fact that, at each node, there are exactly two successor states and one traded risky asset together with the money-market account, which yields a unique replication strategy.
\end{proof}

\subsection{Proof of Theorem~\ref{thm:wellposedness}}
\label{app:proof-wellposedness}

\begin{proof}
Let $\Theta$ be a compact subset of $\mathcal A_\eta(m,M)$.
We prove \Cref{thm:wellposedness} through three steps which allow application of the Weierstrass Extreme Value Theorem:

	\begin{enumerate}
		\item the set of trees $\Tree^{(\theta)}$ generated by the admissible class \(\theta\in\Theta\subset\mathcal A_\eta(m,M)\) is compact,
		
	\item for every strictly admissible base tree \(\Tree^{(\theta_0)}\) with \(\theta_0\in\mathcal A_\eta(m,M)\), there exists \(\varepsilon_0>0\) such that every log-deformation of \(\Tree^{(\theta_0)}\) with sup norm smaller than \(\varepsilon_0\) remains strictly admissible, and
		
		\item 
        for any value of \(\lambda_{\mathrm{prob}}\geq 0\)
		and \(\lambda_{\mathrm{space}}\geq 0\), the penalized objective function
		\(\theta\mapsto\mathcal L(\theta)\) defined in \eqref{eq:complete-objective} 
        %admits at least one		global minimizer on \(\Theta\). 
        is continuous on \(\mathcal A_\eta(m,M)\).
	\end{enumerate}
	
\smallskip
\noindent\emph{Step 1: Compactness of trees $\{\Tree^{(\theta)}:\theta\in\Theta\}$.}

  The set $\Theta$ is compact and $\theta\to  \Tree^{(\theta)}$ is continuous by the assumption on $f_\theta$. This implies that $\Tree^{(\theta)}$ is compact as well.

	\smallskip
\noindent\emph{Step 2: Strict admissibility under small log-deformations.}
		
		Take any $\theta\in\Theta\subset A_\eta(m,M)$. For every \((n,j)\), if we define \( g_n=\exp((r_n-q_n)\Delta t)\), we have due to \eqref{eq:Aeta-margin},
		\begin{equation}\label{eq:2etabound}
		S_{n+1,j+1}^{(\theta)}-S_{n+1,j}^{(\theta)}
		=
		\bigl(S_{n+1,j+1}^{(\theta)}-g_nS_{n,j}^{(\theta)}\bigr)
		+
		\bigl(g_nS_{n,j}^{(\theta)}-S_{n+1,j}^{(\theta)}\bigr)
		\ge 2\eta.
		\end{equation}
		By definition  \eqref{eq:rnp}, the numerator  of \(p_{n,j}^{(\theta)}\) is at least \(\eta\), the denominator minus the numerator is again at least \(\eta\), and the denominator is at most \(M-m\). Therefore
		\[
		p_{n,j}^{(\theta)}\in
		\left[\frac{\eta}{M-m},\,1-\frac{\eta}{M-m}\right].
		\]
		Moreover, by \eqref{eq:2etabound}  and since \(S_{n,j}^{(\theta)}\le M\), the corresponding log-prices satisfy, for $n>1$,
		\begin{equation}\label{eq:logsbounds}
		\log\!\left(\frac{S_{n,j+1}^{(\theta)}}{S_{n,j}^{(\theta)}}\right)
		=
		\log\!\left(1+\frac{S_{n,j+1}^{(\theta)}-S_{n,j}^{(\theta)}}{S_{n,j}^{(\theta)}}\right)
		\ge
		\log\!\left(1+\frac{2\eta}{M}\right),
		\end{equation}

		Take \(\Tree^{(\theta_0)}\) with \(\theta_0\in\mathcal A_\eta(m,M)\), and define
		\[
		\Tree^{(\theta)}_{n,j}= \Tree^{(\theta_0)}_{n,j}e^{\delta_{n,j}}.
		\]
		Note that by \eqref{eq:Aeta-margin} the $g_n$ are bounded, and set
		\begin{equation}\label{eq:CM}
		C_M=M (1+\max_{0\leq n\leq N_T} g_n),
		\qquad
		\varepsilon_0=\log\!\left(1+\frac{\eta}{2C_M}\right).
		\end{equation}
		If \(\|\delta\|_\infty:=\displaystyle\max_{0\leq j\leq n\leq N_T}|\delta_{n,j}|\le \varepsilon_0\) then, by \eqref{eq:logsbounds} and \eqref{eq:CM},
		\begin{equation}\label{eq:newbounda}
		me^{-\varepsilon_0}
		\le
		\Tree^{(\theta)}_{n,j}
		\le
		Me^{\varepsilon_0},
		\end{equation}
		for any \((n,j)\). Using this, and the fact that $\Tree^{(\theta_0)}$ satisfies \eqref{eq:Aeta-bounds}-\eqref{eq:Aeta-margin},
        \begin{align*}
			g_n\Tree^{(\theta)}_{n,j}-\Tree^{(\theta)}_{n+1,j}
			&=
			\bigl(g_n\Tree^{(\theta_0)}_{n,j}-\Tree^{(\theta_0)}_{n+1,j}\bigr)
			+
			g_n\bigl(\Tree^{(\theta)}_{n,j}-\Tree^{(\theta_0)}_{n,j}\bigr)
			-
			\bigl(\Tree^{(\theta)}_{n+1,j}-\Tree^{(\theta_0)}_{n+1,j}\bigr)\\
			&\ge
			\eta
			-
			g_n\bigl|\Tree^{(\theta)}_{n,j}-\Tree^{(\theta_0)}_{n,j}\bigr|
			-
			\bigl|\Tree^{(\theta)}_{n+1,j}-\Tree^{(\theta_0)}_{n+1,j}\bigr|\\
            &\ge
            \eta-(1+\max_{0\leq j\leq n\leq N_T}g_n)M\bigl(e^{\varepsilon_0}-1\bigr)
		=
		\eta-C_M\bigl(e^{\varepsilon_0}-1\bigr)
		\ge
		\frac{\eta}{2}.
		\end{align*}
		A similar argument gives
		\[
		\Tree^{(\theta)}_{n+1,j+1}-g\Tree^{(\theta)}_{n,j}\ge \frac{\eta}{2},
		\]
		and 
        this shows that the deformed tree satisfies the admissibility condition \eqref{eq:Aeta-margin}. 
		In particular, the deformed tree remains arbitrage-free and complete by \Cref{prop:arbfree}.
		
\smallskip
\noindent\emph{Step 3: Continuity of the objective function.}
		
		For each node \((n,j)\), the local risk-neutral probability  is obtained from the three adjacent node values through the map
		in \eqref{eq:rnp} and
		 the denominator is bounded
		below by \(2\eta\) due to \eqref{eq:2etabound}, so the map is Lipschitz in $\theta$ on the compact domain \(\Theta\). Hence each local probability \(p_{n,j}^{(\theta)}\) depends continuously on the tree nodes and hence, by the assumed continuity of $(\theta,x,y)\mapsto f_\theta(x,y)$, it is also continuous in 
        \(\theta\).
        
				This property is then inherited by the Arrow--Debreu state prices $\lambda_{n,j}^{(\theta)}$, which are generated recursively through finitely many additions and multiplications involving these probabilities and the discount factor, and for the prices of
European calls and puts with maturities that correspond to grid points, since their payoff maps \(x\mapsto (x-K)_+\) and \(x\mapsto (K-x)_+\) are continuous. For American calls and puts, the backward induction recursion consists of finitely many continuous operations and pointwise maxima of continuous functions, so the same holds.
		If the maturity for an option lies between two grid dates, the pricing rule uses a convex
		combination of the two neighbouring grid prices; this also preserves continuity. 
        
        We thus conclude that for all options $i$ the map $
        \theta\mapsto\Pi_i^{(\theta)}$, and therefore the $\theta\mapsto \mathrm{MSE}(\theta)$ function defined in \eqref{eq:mse}, is continuous.
		By \eqref{eq:complete-objective} the proof of Theorem~\ref{thm:wellposedness} is therefore finished if we can prove the same for  \(\mathcal{P}_{\mathrm{prob}}\)  and  \(\mathcal{P}_{\mathrm{space}}\).

		For the probability penalty \(\mathcal{P}_{\mathrm{prob}}\) this follows from \eqref{eq:Pprob} since the ReLU function $\phi$ is continuous.
		The spatial penalty \(\mathcal{P}_{\mathrm{space}}\) is built from finitely many continuous
	operations involving node values, clipped risk-neutral probabilities, logarithms, finite
	differences of the tree-implied local variances, and normalizations, so  it is also continuous on \(\mathcal A_\eta(m,M)\). This concludes the proof.
\end{proof}
	
	\FloatBarrier
	\section{Additional local-volatility estimates} 
	\label{app:lv-slices}
	
	For completeness, we report additional comparisons between the true local volatility and the tree-implied local volatility for the two training grids considered in the synthetic-data experiments. These figures complement the surface plots discussed in the main text in \Cref{sec:sy_da_ex}. Figure~\ref{fig:appendix_lv_slices} shows results without regularization, and
Figure~\ref{fig:appendix_lv_slices_lambda} results for $\lambda_{\mathrm{space}}=3$.
	
		As can be seen from Figure~\ref{fig:appendix_lv_slices}, the local volatility extracted node by node from the calibrated tree may display  irregularities, in particular when the model is calibrated on a sparse dataset. The GPR reconstruction acts as an effective smoothing step and yields a more regular local-volatility estimate.
	
A comparison between Figures~\ref{fig:appendix_lv_slices} and
\ref{fig:appendix_lv_slices_lambda} further suggests that the introduction of the penalty term for regularization reduces the amplitude of the 
oscillations.
 This effect is visible in both training configurations, and especially in the sparse $3\times 6$ case. However, once the GPR smoothing step is applied, the additional improvement remains modest. This indicates that the penalty is more useful as a device for guiding the training of the neural network than as a tool to improve the final reconstructed local-volatility surface.
		
	\begin{figure}[h]
		\centering
		\includegraphics[width=0.95\textwidth]{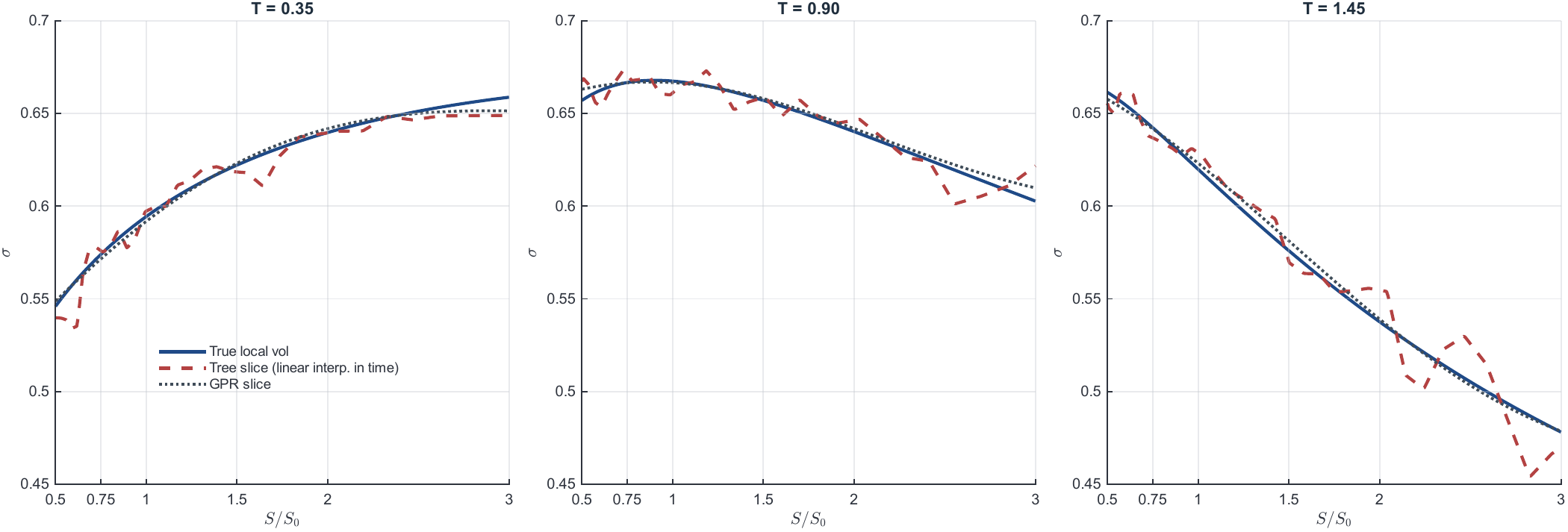}
		\includegraphics[width=0.95\textwidth]{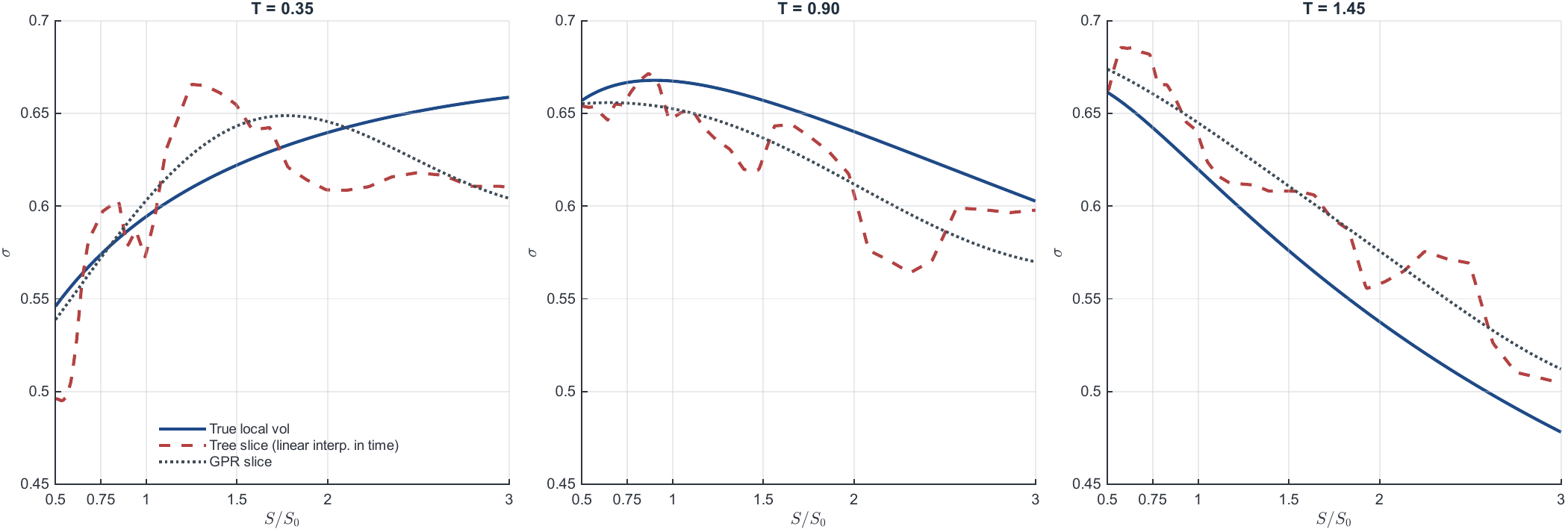}
		\caption{Slice-wise comparison of the true local volatility and the tree-implied local volatility for the \(10\times 20\) training grid (above) and for the \(3\times 6\) training grid (below), at three representative maturities. Parameters: $N_T=720$, $N_E=8000$, $\lambda_{\mathrm{space}}=0$.}
		\label{fig:appendix_lv_slices}
	\end{figure} 
	
		\begin{figure}[h]
		\centering
		\includegraphics[width=0.95\textwidth]{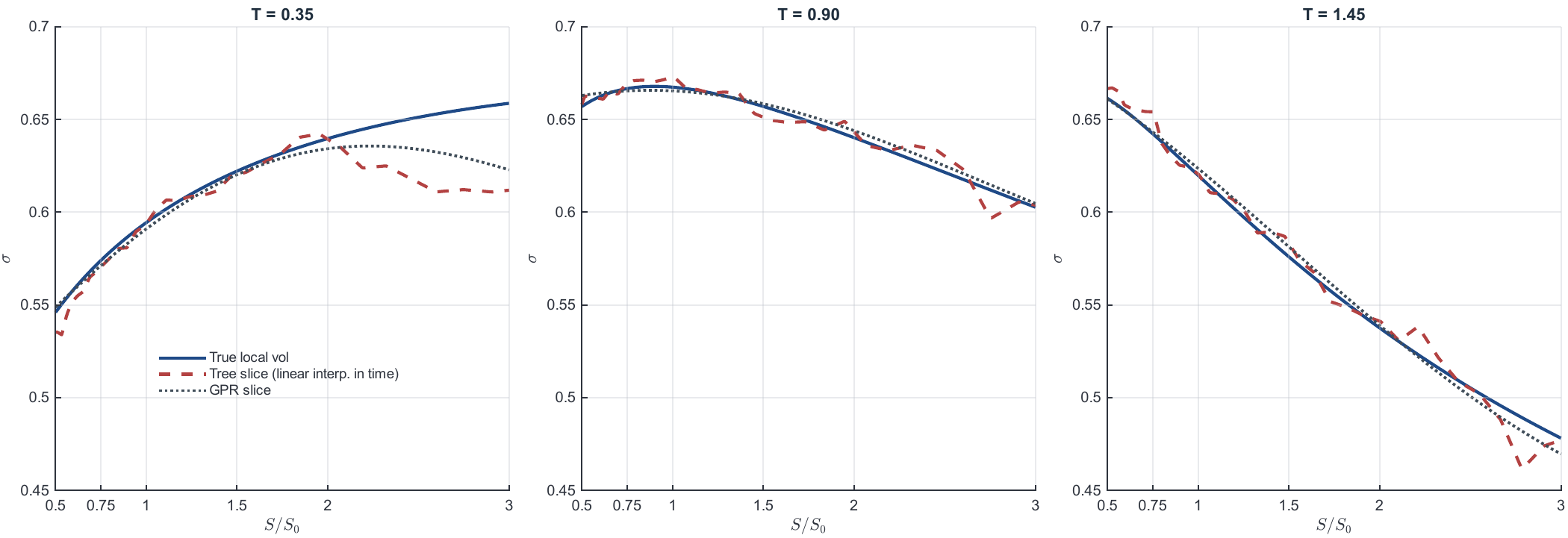}
		\includegraphics[width=0.95\textwidth]{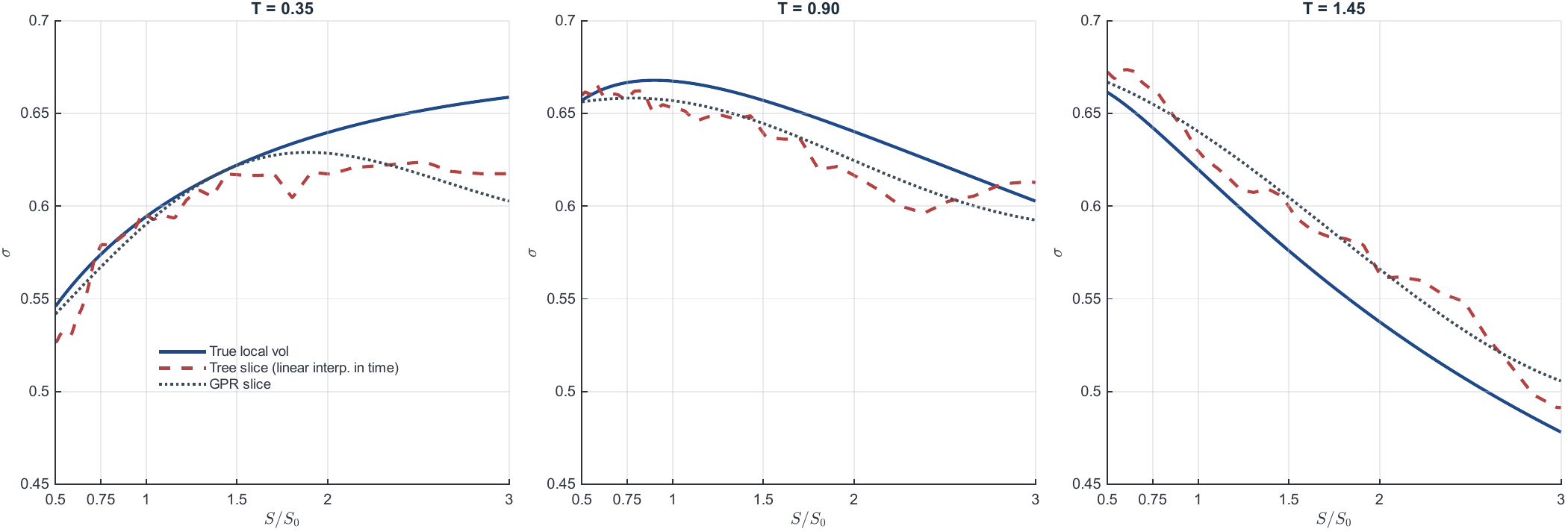}
		\caption{Slice-wise comparison of the true local volatility and the tree-implied local volatility for the \(10\times 20\) training grid (above) and for the \(3\times 6\) training grid (below), at three representative maturities. Parameters: $N_T=720$, $N_E=8000$, $\lambda_{\mathrm{space}}=3$.}
		\label{fig:appendix_lv_slices_lambda}
	\end{figure} 

	% \FloatBarrier

	\FloatBarrier
	
	\bibliographystyle{plainnat}
	\bibliography{My_biblio}
\end{document}